\documentclass[11pt]{article}
\pdfoutput=1 
\usepackage[sc]{mathpazo}

\usepackage{palatino,microtype}
\usepackage[margin=15pt,font=small,labelfont={bf}]{caption}
\usepackage[margin=1in]{geometry}
\usepackage[english]{babel}
\usepackage[utf8]{inputenc}
\usepackage[compact]{titlesec}
\usepackage{cmap}
\usepackage[T1]{fontenc}
\usepackage{bm}
\usepackage{booktabs}
\usepackage{mathtools}
\usepackage{amsfonts}
\usepackage{amsmath}
\usepackage{amssymb}
\usepackage{amsthm}
\usepackage{float}
\usepackage{graphics}
\usepackage[numbers]{natbib}
\usepackage{hyperref}
\usepackage[svgnames,table]{xcolor}
\hypersetup{colorlinks={true},urlcolor={blue},linkcolor={DarkBlue},citecolor=[named]{DarkGreen},linktoc=all}

\usepackage[capitalise,nameinlink,noabbrev]{cleveref}
\usepackage{doi}

\usepackage{colortbl}
\colorlet{myred}{red!25}
\colorlet{myblue}{blue!25}
\colorlet{mygreen}{green!25}

\newcommand{\BibTeX}{\rm B\kern-.05em{\sc i\kern-.025em b}\kern-.08em\TeX}

\usepackage[labelfont={normalfont,bf},textfont=it]{caption}
\usepackage{subcaption}

\usepackage{nicefrac}
\usepackage{pifont}

\usepackage{apxproof}
\usepackage{array,multirow,graphicx,bigdelim}

\input{insbox}
\usepackage{tikz}
\usepackage{tikz-dimline}
\usepackage{tikz-cd}
\usetikzlibrary{fit,calc,math,shapes,shapes.multipart,decorations.text,arrows,decorations.markings,decorations.pathmorphing,shapes.geometric,positioning,decorations.pathreplacing}
\tikzset{snake it/.style={decorate, decoration=snake}}

\colorlet{mygray}{gray!40}

\makeatletter
\let\oldnl\nl
\newcommand{\nonl}{\renewcommand{\nl}{\let\nl\oldnl}}
\makeatother

\usepackage{thmtools,thm-restate}
\usepackage[capitalise,nameinlink,noabbrev]{cleveref}

\newtheorem{lemma}{Lemma}
\newtheorem{theorem}{Theorem}

\newtheorem{proposition}{Proposition}

\Crefname{claim}{Claim}{Claims}

\Crefname{claim}{Claim}{Claims}
\Crefname{corollary}{Corollary}{Corollaries}
\Crefname{definition}{Definition}{Definitions}
\Crefname{example}{Example}{Examples}
\Crefname{lemma}{Lemma}{Lemmas}
\Crefname{property}{Property}{Properties}
\Crefname{proposition}{Proposition}{Propositions}
\Crefname{remark}{Remark}{Remarks}
\Crefname{theorem}{Theorem}{Theorems}

\newcommand{\bbN}{\mathbb{N}}

\newcommand{\argmax}{\mbox{argmax}}

\newcommand{\MMS}{\textrm{\textup{MMS}}}
\newcommand{\NP}{\textrm{\textup{NP}}}
\newcommand{\NPC}{\textrm{\textup{NP-complete}}}
\newcommand{\NPc}{\textrm{\textup{NP-c}}}
\newcommand{\NPh}{\textrm{\textup{NP-h}}}
\newcommand{\NPH}{\textrm{\textup{NP-hard}}}
\newcommand{\EF}[1]{\ifstrempty{#1}{\textrm{\textup{EF}}}
{\textrm{\textup{EF{$#1$}}}}}

\newcommand{\RainbowColoring}{\textsc{Rainbow Coloring}}
\newcommand{\Partition}{\textsc{Partition}}
\newcommand{\pCompletion}{\textsc{$p$-Completion}}
\newcommand{\EFoneCompletion}{\textsc{\EF{1}-Completion}}

\newcommand{\Prop}[1]{\ifstrempty{#1}{\textrm{\textup{Prop}}}{\textrm{\textup{Prop{$#1$}}}}}

\newcommand{\ProponeCompletion}{\textsc{\Prop{1}-Completion}}
\newcommand{\MMSCompletion}{\textsc{\MMS-Completion}}
\newcommand{\POCompletion}{\textsc{\PO{}-Completion}}
\newcommand{\EFonePOCompletion}{\textsc{\EF{1}+\PO{}-Completion}}
\newcommand{\ProponePOCompletion}{\textsc{\Prop{1}+\PO{}-Completion}}
\newcommand{\MMSPOCompletion}{\textsc{\MMS{}+\PO{}-Completion}}

\newcommand{\PO}{\textup{PO}}
\newcommand{\Polytime}{\textup{Poly}}
\newcommand{\PseudoPolytime}{\textup{Pseudopoly}}
\newcommand{\V}{\mathcal{V}}
\renewcommand{\>}{\succ}

\usepackage{algorithm}
\usepackage{algpseudocode}

\theoremstyle{remark}

\let\displaystyle\textstyle

\title{Fair and Efficient Completion of Indivisible Goods}

\author{
	\begin{tabular}{m{0.12\textwidth}m{0.12\textwidth}m{0.12\textwidth}m{0.12\textwidth}
 }
		\multicolumn{2}{c}{\textbf{Vishwa Prakash HV}} & \multicolumn{2}{c}{\textbf{Ayumi Igarashi}} 
        \\
		\multicolumn{2}{c}{Chennai Mathematical Institute} & \multicolumn{2}{c}{The University of Tokyo} 
        \\ 
		\multicolumn{2}{c}{\href{mailto:vishwa@cmi.ac.in}{\small{\texttt{vishwa@cmi.ac.in}}}} & \multicolumn{2}{c}{\href{mailto:igarashi@mist.i.u-tokyo.ac.jp}{\small{\texttt{igarashi@mist.i.u-tokyo.ac.jp}}}}
        \\
        &&&\\
		\multicolumn{4}{c}{\textbf{Rohit Vaish}}\\
            \multicolumn{4}{c}{Indian Institute of Technology Delhi}\\
            \multicolumn{4}{c}{\href{mailto: rvaish@iitd.ac.in}{\small{\texttt{ rvaish@iitd.ac.in}}}}
	\end{tabular}
}

\date{}

\begin{document}

\maketitle 
\begin{abstract}
    We formulate the problem of \emph{fair and efficient completion} of indivisible goods, defined as follows: Given a partial allocation of indivisible goods among agents, does there exist an allocation of the remaining goods (i.e., a completion) that satisfies fairness and economic efficiency guarantees of interest? We study the computational complexity of the completion problem for prominent fairness and efficiency notions such as envy-freeness up to one good (\EF{1}), proportionality up to one good (\Prop{1}), maximin share (\MMS{}), and Pareto optimality (\PO{}), and focus on the class of additive valuations as well as its subclasses such as \emph{binary additive} and \emph{lexicographic} valuations. We find that while the completion problem is significantly harder than the \emph{standard} fair division problem (wherein the initial partial allocation is empty), the consideration of restricted preferences facilitates positive algorithmic results for threshold-based fairness notions~(\Prop{1} and \MMS{}). On the other hand, the completion problem remains computationally intractable for envy-based notions such as \EF{1} and \EF{1}+\PO{} even under restricted preferences.
\end{abstract}

\section{Introduction}\label{sec:introduction}
The problem of fair and efficient allocation of resources, particularly discrete (or indivisible) resources, is fundamental to economics and computer science. Such problems arise in a multitude of applications such as course allocation at universities~\citep{BCK+17course}, inheritance division~\citep{GP15spliddit}, dividing household chores~\citep{IY23kajibuntan}, and dispute resolution~\citep{BT96fair}. There is now a well-developed theory on discrete fair division spanning various notions of fairness and algorithmic results~\citep{AAB+23fair}.

The existing theory of fair allocation of indivisible resources assumes that all resources are initially \emph{unassigned}. While this may be a natural assumption, it is not always applicable in practical scenarios where resource allocation is predetermined or specific agents are designated to handle particular tasks. For instance, in a heterosexual couple, certain household tasks like breastfeeding can only be assigned to a specific individual. In course allocation, certain courses may be mandatory for a particular group of students. Similarly, 
in dispute resolution, legal agreements may specify allocating specific assets to each party. 
In such settings, it is more natural to assume that specific resources have been \emph{pre-assigned} to particular agents. We refer to such resources as \emph{frozen}. The goal is to find a fair and efficient \emph{completion} of the frozen allocation; that is, to allocate the unassigned resources such that the overall allocation is fair and efficient.

Frozen resources present an interesting challenge, as illustrated by the following example. Consider a scenario where two agents are each pre-allocated one good, and there is one unallocated precious good. If these agents envy each other under the frozen allocation, then no completion can satisfy envy-freeness up to one good (\EF1{}). An allocation is \EF{1} if any pairwise envy can be eliminated by removing some good from the envied bundle~\citep{B11combinatorial}. In the absence of frozen goods, an \EF{1} allocation always exists in the standard setting of fair division~\citep{LMM+04approximately,CKM+19unreasonable}. However, as the above example demonstrates, an \EF{1} \emph{completion} may not exist with frozen resources. Therefore, studying the existence and computation of fair and efficient completion of indivisible resources is crucial.

Another challenge posed by the completion problem is in the \emph{formulation} of solution concepts. Specifically, let us consider the maximin share (\MMS{}) criterion~\citep{B11combinatorial}, where fairness is defined based on a utility threshold for each agent. This threshold is defined in terms of an $n$-person cut-and-choose game, where the cutter is the last to pick a bundle. When some resources are preassigned, there can be multiple ways of extending this definition. For example, the cutter may only consider the extension of its own frozen bundle, or it may seek to maximize the value of the least-valued bundle overall. These choices lead to different existential and computational results. Thus, adapting existing solution concepts to the completion setting is non-trivial.

\renewcommand{\arraystretch}{1.1}
\begin{table*}
\centering

\scriptsize
\footnotesize
\begin{tabular}{|c|c|c|c|c|c|c|}
 \hline
 \multirow{3}{*}{Problem $\downarrow$} & \multicolumn{2}{c|}{\textbf{Binary Additive}} & \multicolumn{2}{c|}{\textbf{Lexicographic}} & \multicolumn{2}{c|}{\textbf{General Additive}}\\
 \cline{2-7}
 &Standard & Completion & Standard & Completion & Standard & Completion\\
 &Problem & Problem & Problem & Problem & Problem & Problem\\
 \hline
 {\EF{1}} & {\Polytime{}\(^\dag\)} & \cellcolor{myred} \NPc{}~(Th.~\ref{thm:EFoneBinaryAdditiveNPhard}) 
 &{\Polytime{}\(^\dag\)} & \cellcolor{myred} \NPc{}~(Th.~\ref{thm:EF1-Lexoicographic-NPc}) 
 & {\Polytime{}\(^\dag\)} & \cellcolor{myred} \NPc{}~(Th.~\ref{thm:EFoneTwoAgentAdditiveNPhard},\ref{thm:EFoneThreeAgentIdenticalAdditiveNPhard})\\
 \hline
 {\Prop{1}}  & {\Polytime{}\(^\ddag\)} & \cellcolor{mygreen}\Polytime{}~(Th.~\ref{thm:MMS:PROP1:binary}) 
 & {\Polytime{}\(^\ddag\)} & \cellcolor{mygreen} \Polytime{}~(Th.~\ref{thm:PROP1-PO-Lex-Poly})
 &
 {\Polytime{}\(^\ddag\)} & \cellcolor{myred} \NPc{}~(Th.~\ref{thm:EFoneTwoAgentAdditiveNPhard},\ref{thm:EFoneThreeAgentIdenticalAdditiveNPhard})\\
 \hline
 {\MMS{}} &  {\Polytime{}\(^\S\)}  & \cellcolor{mygreen}  \Polytime{}~(Th.~\ref{thm:MMS:PROP1:binary}) 
 & {\Polytime{}\(^\diamondsuit\)}   & \cellcolor{mygreen} {\Polytime{}}~(Th.~\ref{thm:decideMMS:lexicographic}) 
 & ${\Sigma^2_p}^\S$ & \cellcolor{myred} \NPh{}~(Th.~\ref{thm:MMSTwoAgentAdditiveNPhard})
 \\
 %
 %
 \hline
 {\EF{1}+\PO{}} 
 & {\Polytime{}\(^\star\)} & \cellcolor{myred} \NPc{}~(Cor.~\ref{cor:EFonePOBinaryAdditiveNPhard}) 
 &{\Polytime{}\(^\diamondsuit\)}
 &?
 & {\PseudoPolytime{}\(^\triangle\)} & \cellcolor{myred}\NPh{}~(Cor.~\ref{cor:EFonePOThreeAgentAdditiveNPhard})\\
 %
 %
 \hline
 {\Prop{1}+\PO{}} 
 & {\Polytime{}\(^\P\)} & \cellcolor{mygreen}\Polytime{}~(Cor.~\ref{thm:PO:MMS:PROP1:binary}) 
 &  {\Polytime{}\(^\P\)} & \cellcolor{mygreen}\Polytime{}~(Th.~\ref{thm:PROP1-PO-Lex-Poly}) 
 & {\Polytime{}\(^\P\)} & \cellcolor{myred} \NPh{}~(Cor.~\ref{cor:EFonePOThreeAgentAdditiveNPhard})\\
 %
 %
 \hline
 {\MMS{}+\PO{}} 
 & {\Polytime{}\(^\S\)}  & \cellcolor{mygreen} \Polytime{}~(Cor.~\ref{thm:PO:MMS:PROP1:binary}) 
 & {\Polytime{}\(^\diamondsuit\)} 
 & ? 
 & ? & ?\\
 %
 %
 \hline
\end{tabular}

\vspace{0.1in}
\caption{Summary of our results for indivisible goods. Each row corresponds to a fairness/efficiency property. 
Each column corresponds to a preference restriction and whether we are looking at the standard or the completion problem; `standard' is a special case of completion problem when there are no frozen goods. Each cell shows the computational complexity of the standard or completion problem under specific scenarios. Entries marked with \(\dag\) is due to \cite{CKM+19unreasonable}, \(\ddag\) is due to \cite{CFS17fair}, \(\S\) is due to \cite{BL16characterizing}, \(\P\) is due to \cite{AMS20polynomial}, \(\diamondsuit\) due to \cite{HSV+21fair}, \(\star\) due to \cite{BKV18greedy} and \(\triangle\) is due to \cite{BKV18finding}}
\label{tab:Results}

\end{table*}

\paragraph{Our Contributions.}
The main conceptual contribution of our work is the formulation of the model of \emph{fair and efficient completion} of indivisible goods. In our model, the assignment of certain resources is \emph{fixed}, while the remaining resources can be flexibly assigned.
We study the computational complexity of the completion problem for various notions of \emph{fairness}, including envy-freeness up to one good (\EF{1}), proportionality up to one good (\Prop{1}), and maximin share (\MMS{}), as well as their combinations with \emph{economic efficiency} (specifically, Pareto optimality). Our results, summarized in \Cref{tab:Results}, span the well-studied class of \emph{additive} valuations as well as its subclasses such as \emph{binary additive} and \emph{lexicographic} valuations. Some of the technical highlights of our work are:
\begin{itemize}
    \item 
    For \emph{binary additive} valuations, we show that achieving an \EF{1} completion is computationally intractable even with the restrictiveness of preferences. This highlights a stark contrast with the standard setting, where an \EF{1} allocation can be computed efficiently~\cite{LMM+04approximately,CFS17fair}. However, for \Prop{1} and \MMS{} notions (and their combinations with Pareto optimality), we demonstrate that the completion problem admits efficient algorithms. Moreover, we establish that if the initial allocation (i.e., allocation of frozen goods) is Pareto optimal, the existence of \MMS{}+\PO{} is guaranteed. 
    \item For \emph{lexicographic} valuations, we again encounter hardness for \EF{1}. This class of valuations is orthogonal to the class of binary additive valuations; here, each agent's preference over bundles is determined by the most preferred good in the set difference of the bundles. Nevertheless, we show that the completion problem for \Prop{1}, \Prop{1}+\PO{}, and \MMS{} is tractable, highlighting the difference between threshold-based notions and envy-based fairness notions. The other completion problems (\MMS{}+\PO{} and \EF{1}+\PO{}) seem challenging, and we leave them as open problems.

    \item For general \emph{additive} valuations, we show that the completion problem for \EF{1} and \Prop{1} (and their combinations with Pareto optimality) is \NPH{}. This hardness result holds even for two agents with additive valuations and for three agents with identical additive valuations. We also establish the hardness of deciding the existence of an \MMS{} allocation, a problem whose lower bound on complexity is open in the standard setting. Moreover, we demonstrate that $\alpha$-\MMS{} may not exist for every $\alpha>0$ even when the initial allocation is Pareto optimal. 
\end{itemize}
Notably, our hardness results for general additive and binary additive valuations hold even when the initial allocation is Pareto optimal.

\paragraph{Related Work.}
The idea of ``completion'' has been studied in computational social choice~\citep{BCE+16handbook}, although under different motivation and technical assumptions than ours. In the voting literature, the notion of \emph{possible} and \emph{necessary winners}~\citep{KL05voting,XC11determining} models completion of incomplete preferences (as opposed to allocations). In the fair division literature, the notions of \emph{possible} and \emph{necessary fairness} consider agents with ordinal preferences over the goods and examine all realizations of cardinal valuations that are consistent with these preferences~\citep{ASW23computational,ALX+23possible}.

Our approach of assuming frozen assignments is similar to prior work by Dias et al.~\cite{DDD+03stable} on fixed and forbidden pairs in the stable matching problem. Studying forbidden assignments in fair division is a relevant direction for future work.

We study binary additive and lexicographic valuations, which are subclasses of additive valuations. These subclasses have been extensively studied in the fair division literature and have facilitated several positive existential and algorithmic results for various fairness and efficiency notions~\citep{BarmanKrishna2022IJCAI,barman2021truthful,benabbou2020,halpern2020fair,babaioff2020fair,HSV+20fairAAAI,HSV+21fair,HSVL23,HAW23IJCAI}. For example, a maximin share (\MMS{}) allocation always exists and can be efficiently computed under binary additive and lexicographic valuations while it might fail to exist for additive valuations~\citep{KP18fair}.

Our work also connects with the literature on fair division with \emph{subsidy}~\citep{HS19fair}, which compensates agents with a divisible resource (e.g., money) to eliminate envy. The unallocated indivisible goods in our model play a similar role as subsidy, but our focus is on achieving approximate---rather than exact---envy-freeness alongside Pareto optimality.

\section{Preliminaries}\label{sec:preliminaries}
For any natural number $s \in \bbN$, let $[s] \coloneqq \{1,2,\ldots,s\}$.

\paragraph{Problem instance.}
An instance of the \emph{fair completion} problem is given by the tuple $\langle N, M, \V, F \rangle$ where $N=[n]$ denotes a set of $n$ {\em agents}, $M=[m]$ denotes a set of $m$ indivisible goods, $\V = (v_1,v_2,\dots,v_n)$ denotes a \emph{valuation profile}, and $F = (F_1,F_2,\dots,F_n)$ denotes the partial allocation of \emph{frozen} goods (explained below).

\paragraph{Valuation function.}
The valuation profile $\V$ consists of the \emph{valuation functions} $v_1,\dots,v_n$, where $v_i: 2^M \rightarrow \mathbb{R}_+$ specifies the value that agent $i$ has for any subset of goods; here, $\mathbb{R}_+$ is the set of non-negative reals. Each agent values the empty set at $0$, i.e., for every $i \in N$, $v_i(\emptyset)=0$. We will assume that valuations are \emph{additive}, which means that the value of any subset of goods is equal to the sum of values of the individual goods in that set. That is, for any $i\in N$ and $S \subseteq M$, $v_i(S) \coloneqq \sum_{g\in S} v_i(\{g\})$. For simplicity, we will write $v_i(g)$ to denote $v_i(\{g\})$.

We will focus on two subclasses of additive valuations: \emph{binary additive} and \emph{lexicographic}.

\emph{Binary additive valuations}: We say that agents have binary additive valuations if for every agent $i \in N$ and every good $g \in M$, we have $v_{i}(g) \in \{0,1\}$. An agent $i$ \emph{approves} a good $g$ whenever $v_{i}(g)=1$.

\emph{Lexicographic valuations:} We say that agents have lexicographic valuations if each agent prefers any bundle containing its favorite good over any bundle that doesn't, subject to that, it prefers any bundle containing its second-favorite good over any bundle that doesn't, and so on. Formally, each agent $i$ is associated with a strict and complete ranking $\>_i$ over the goods in $M$. Given any two distinct bundles $X$ and $Y$, agent $i$ prefers $X$ over $Y$ if and only if there exists a good $g \in X \setminus Y$ such that any good in $Y$ that is more preferable than $g$ is also contained in $X$, i.e., $\exists g \in X \setminus Y$ such that $\{g' \in Y: g' \>_i g\} \subseteq X$. Note that the lexicographic preference extension induces a total order over the bundles in $2^M$. Also, 
lexicographic preferences can be specified using only an \emph{ordinal} ranking over the goods~\citep{HSV+21fair}.

\paragraph{Allocation.}
A partial \emph{allocation} of the goods in $M$ is an ordered subpartition $A=(A_1,\ldots, A_n)$ of $M$, where, for all distinct $i,j \in N$, $A_i \cap A_j = \emptyset$ and $\cup_{i \in N} A_i \subseteq M$. The subset $A_i$ is called the \emph{bundle} of agent~$i$. An allocation is \emph{complete} if all goods are allocated, i.e., $\cup_{i\in N} A_i=M$. For simplicity, we will use the term `allocation' to refer to a complete allocation and write `partial allocation' otherwise.

\paragraph{Frozen allocation and completion.}
A \emph{frozen} allocation $F = (F_1,\dots,F_n)$ is a partial allocation of $M$ that denotes the subset of goods that are pre-assigned to the agents; that is, the goods in $F_i$ are irrevocably allocated to agent $i$ to begin with. We will denote by $U \coloneqq M \setminus \cup_{i \in N} F_i$ the set of \emph{unallocated} goods, i.e., the goods that are not pre-assigned. A \emph{completion} $C = (C_1,C_2,\dots,C_n)$ refers to an allocation of the unallocated goods, i.e., $\cup_{i \in N} C_i = U$ and $C_i \cap C_j = \emptyset$ for all distinct $i,j \in N$. Define an $n$-tuple $A \coloneqq (A_1,\dots,A_n)$ such that for every $i \in N$, $A_i \coloneqq F_i \cup C_i$. Notice that $A$ is a complete allocation of the goods in $M$. We will say that allocation $A$ \emph{completes} the frozen allocation $F$.

\paragraph{Fairness notions.}
A partial allocation $A = (A_1,\dots,A_n)$ is said to be \emph{envy-free} (\EF{}) if for every pair of agents $i,j \in N$, we have $v_i(A_i) \geq v_i(A_j)$~\citep{GS58puzzle,F67resource}, and \emph{envy-free up to one good} (\EF{1}) if for every pair of agents $i,j \in N$ such that $A_j \neq \emptyset$, there exists some good $g \in A_j$ such that $v_i(A_i) \geq v_i(A_j \setminus \{g\})$~\citep{B11combinatorial,LMM+04approximately}. We say that agent $i$ has \emph{\EF{1}-envy} towards agent $j$ if for every good $g \in A_j$, we have $v_i(A_i) < v_i(A_j \setminus \{g\})$.

A partial allocation $A$ is \emph{proportional} (\Prop{}) if for each agent $i \in N$, we have $v_i(A_i) \geq \frac{1}{n} \cdot v_i(M)$~\citep{S48problem}, and \emph{proportional up to one good} (\Prop{1}) if for each agent $i \in N$, there exists a good $g \in \cup_{k \neq i} A_k$ such that $v_i(A_i \cup \{g\}) \geq \frac{1}{n} \cdot v_i(M)$~\citep{CFS17fair}.

A partial allocation $A$ is called \emph{maximin fair} if every agent receives a bundle of value at least its \emph{maximin share} (\MMS{}) value. To define the \MMS{} value of agent $i$, consider any frozen allocation $F$. 
The \emph{\MMS{} partition} of agent $i$ is an allocation, say $A$, that completes $F$ such that the value of the least-valued bundle in $A$ is maximized (according to agent $i$). Observe that agent $i$ evaluates \emph{all} bundles of $A$ under this definition, not just the one that contains $F_i$. The value of the least-valued bundle under the \MMS{} partition is called the \MMS{} value of agent $i$. Formally, the \MMS{} value $\mu_i$ of agent $i$ is defined as
\[
\mu_i \coloneqq  \max_{(C_1,\dots,C_n) \in \Pi_U}  \min_{j \in [n]} v_i(F_j \cup C_j),
\]
where $\Pi_U$ denotes the set of all ordered $n$-partitions of the set of unallocated goods $U$ (the individual bundles in the partition can be empty). 
An allocation $A$ satisfies the \emph{maximin fair share} criterion (also denoted by \MMS{}) if each agent's utility under $A$ is at least its \MMS{} value, i.e., for each agent $i$, $v_i(A_i) \geq \mu_i$. Note that the frozen goods in the bundle $F_j$ for $j \neq i$ are not \emph{actually} assigned to agent $i$; the definition of \MMS{} value $\mu_i$ considers a thought experiment from agent $i$'s perspective wherein it wants to maximize the value of the least-valued bundle by adding only the unallocated goods in $U$. For any given $\alpha \geq 0$, we say that partial allocation $A$ is $\alpha$-\emph{maximin fair} ($\alpha$-\MMS{}) if every agent receives a bundle of value at least $\alpha$ times its \MMS{} value, i.e., for each agent $i$, $v_i(A_i) \geq \alpha \cdot \mu_i$.

\paragraph{Pareto optimality.}
A partial allocation $X$ is said to Pareto dominate another partial allocation $Y$ if for every agent $i$, $v_i(X_i) \geq v_i(Y_i)$ and for some agent $k$, $v_k(X_k) > v_k(Y_k)$. A partial allocation is said to be \emph{Pareto optimal} (\PO{}) if no other partial allocation Pareto dominates it. Note that we allow for the (hypothetical) redistribution of the frozen goods when considering Pareto domination.\footnote{One could consider a weaker form of Pareto optimality by defining Pareto domination only by allocations that complete the given frozen allocation. For additive valuations, this notion is equivalent to the Pareto optimality with respect to only the unallocated goods. Our algorithmic and hardness results continue to hold for this notion.
}

We will also be interested in Pareto optimality of only the \emph{frozen} allocation. A frozen allocation $F = (F_1,\dots,F_n)$ is said to be Pareto optimal (with respect to the frozen goods in $\cup_{i \in N} F_i$) if there is no other allocation of the frozen goods that Pareto dominates $F$.

\paragraph{\pCompletion{}.} Let $p$ denote any fairness property (e.g., \EF{1}, \Prop{1}, \MMS{}) or a combination of fairness and efficiency properties (e.g., \EF{1}+\PO{}). The computational problem \pCompletion{} takes as input an instance $\langle N, M, \V, F \rangle$ of the completion problem and asks whether there exists a completion $C$ of the frozen allocation $F$ such that the final allocation $F \cup C$ satisfies the property $p$.

The omitted proofs can be found in Appendix.
\section{Binary Additive Valuations}\label{sec:binary}
We will start the discussion of our technical results with a restricted subclass of additive valuations called \emph{binary additive} valuations wherein, for each agent $i \in N$ and each good $g \in M$, we have $v_{i}(g) \in \{0,1\}$.

\paragraph{Envy-Freeness Up to One Good.}
In the standard setting of fair division (i.e., without any frozen goods), an \EF1{} allocation always exists and can be efficiently computed by a simple round-robin algorithm~\citep{CKM+19unreasonable}. However, as discussed in the Introduction, an \EF1{} allocation could fail to exist when some goods are pre-allocated. 
Our first main result in this section is that \EFoneCompletion{} is computationally intractable even when the agents have binary additive valuations and even when the frozen allocation is Pareto optimal. Note that, for binary additive valuations, an allocation is Pareto optimal if and only if each good is allocated to an agent who approves it. 

\begin{restatable}{theorem}{EFoneBinaryAdditiveNPhard}
For binary additive valuations, \EFoneCompletion{} is \NPC{} even when the frozen allocation is Pareto optimal.
\label{thm:EFoneBinaryAdditiveNPhard}
\end{restatable}

Intuitively, the construction in \Cref{thm:EFoneBinaryAdditiveNPhard} creates ``enough'' envy in the frozen allocation such that compensating for it amounts to solving a hard problem. Observe that binary valuations significantly limit the design of the reduced instance, as the agents can only assign a value of 0 or 1 to the goods. Because of this limitation, it is not immediately clear whether reductions similar to those used for general additive valuations (such as from \Partition{} problem) can be applied here. Our proof circumvents this challenge by reducing from the {\sc Equitable Coloring} problem.

Formally, given an undirected graph $G = (V,E)$, a \emph{proper coloring} is an assignment $\pi \colon V \rightarrow [k]$ of colors to the vertices of $G$ where no two adjacent vertices have the same color. Such a coloring is called \emph{equitable} if the number of vertices in any two color classes are the same. Given an undirected graph $G$ and a positive integer $k$, {\sc Equitable Coloring} asks whether the graph $G$ has an equitable coloring with a given number $k$ of colors. This problem was used by~Hosseini et al.~\cite[Section 7.1]{HSV+19fairarXiv} to show that, in the standard fair division setting, checking the existence of an envy-free allocation is \NPC{} under binary valuations.

\begin{proof}[Proof of \Cref{thm:EFoneBinaryAdditiveNPhard}]
Checking whether a given completion is \EF{1} is clearly in \NP{}. We present a reduction from {\sc Equitable Coloring} to show \NPH{}ness.

Consider an instance $(G,k)$ of {\sc Equitable Coloring}, where $G=(V,E)$ is an undirected graph and $k$ is a positive integer. Assume, without loss of generality, that there are $n=k \cdot t$ vertices for some positive integer $t$. We construct an instance of our problem as follows. 

For each edge $e=\{u,v\} \in E$, we create the edge agent~$e$. 
Additionally, we create color agents $c_1,c_2,\ldots,c_k$ and  
a special agent $s$. 
For each vertex $v \in V$, we create a vertex good $v$. 
Additionally, there are $(t+1)$ dummy goods $d_1,\ldots,d_{t+1}$. 
Each color agent $c_i$ approves all the vertex goods and the $t+1$ dummy goods $d_1,\ldots,d_{t+1}$. 
Each edge agent $e$ approves the goods $u,v$ corresponding to the vertices incident to $e$ only. 
The special agent $s$ approves all the dummy goods $d_1,\ldots,d_{t+1}$ only.
In a frozen allocation $F$, the special agent $s$ receives the dummy goods $d_1,\ldots,d_{t+1}$ while the other agents receive nothing. Observe that $F$ is Pareto optimal.

To show the equivalence of the solutions, let us first suppose that there exists an \EF{1} allocation $A$ that completes $F$. We will show that there exists an equitable coloring $\pi \colon V \rightarrow [k]$. 
Now create an assignment $\pi \colon V  \rightarrow [k]$ where $\pi(v)$ is the color class $i$ such that $v \in A_{c_i}$. Observe that each color agent $c_i$ needs to receive at least $t$ vertex goods under the allocation $A$ since otherwise, it would have \EF{1}-envy towards the special agent $s$. 
Since there are $n=k \cdot t $ vertices, all vertex goods are allocated to color agents. Thus, $|\pi^{-1}(i)|=t$ for each $i \in [k]$. Further, no pair of vertices incident to each edge $e$ is allocated to $c_i$; otherwise agent $e$ would have \EF{1}-envy towards $c_i$ since $e$ receives none of the goods under $A$. Thus, $\pi$ is a proper coloring. We thus conclude that $\pi$ is an equitable coloring.  

Conversely, suppose that $G$ admits an equitable coloring $\pi \colon V \rightarrow [k]$. Note that since $\pi$ is equitable, $|\pi^{-1}(i)| =t$ for each $i \in [k]$. Then construct an allocation $A$ such that agent $s$ receives the dummy goods $d_1,\ldots,d_{t+1}$, each color agent $c_i$ receives exactly $t$ vertex agents $v$ from $\pi^{-1}(i)$, and each edge agent $e$ receives nothing. 

This allocation $A$ completes $F$. We claim that $A$ is \EF{1}. 
Indeed, it is easy to see that the special agent does not envy other agents. Each color agent obtains a value of at least $t$ and hence does not have \EF{1}-envy towards the special agent $s$. Each color agent envies none of the other agents. 

Now consider any edge agent $e$. We know that since $\pi$ is a proper coloring, no pair of vertices incident to $e$ is allocated to the same color agent $c_i$. Thus, from $e$'s perspective, the valuation of every other agent's bundle is at most $1$. This implies that $A$ is \EF{1}, as desired. 
\end{proof}

To obtain the stronger property combination of \EF{1} and \PO{}, a natural approach is to consider a completion that maximizes \emph{Nash welfare}~\citep{CKM+19unreasonable}. More precisely, a completion $C=(C_1,\ldots,C_n)$ is a maximum Nash welfare (MNW) completion if it maximizes the number of agents, say $n$, receiving a positive utility, and subject to that, the geometric mean of their utilities, i.e., $\argmax_C \left( \prod_{i \in N: v_i(F_i \cup C_i)>0}v_i(F_i \cup C_i) \right)^{1/n}$. 

If there is no frozen good, then a maximum Nash welfare allocation is \EF{1} and \PO{}. For binary additive valuations, a maximum Nash welfare allocation can be computed in polynomial time~\citep{DS15maximizing,BKV18greedy}. %
By contrast, a completion that maximizes Nash welfare could fail to be \EF{1} even when an \EF{1}+\PO{} completion exists; see Proposition~\ref{prop:MNW} in Appendix~\ref{appendix:binary}.
Furthermore, the \EFonePOCompletion{} problem turns out to be \NPC{}. This follows from the construction of Theorem~\ref{thm:EFoneBinaryAdditiveNPhard} where each item is allocated to an agent who approves it.

\begin{restatable}{corollary}{EFonePOBinaryAdditiveNPhard}\label{cor:EFonePOBinaryAdditiveNPhard}
\EFonePOCompletion{} is \NPC{} even for agents with binary additive valuations.
\end{restatable}

\paragraph{Proportionality Up to One Good and Maximin Share.}

In contrast to \EF{1}, the completion problem can be efficiently solved for \Prop{1} and \MMS{} notions under binary additive valuations.

\begin{restatable}{theorem}{BinaryMMSPROPone}\label{thm:MMS:PROP1:binary}
For binary additive valuations, \MMSCompletion{} and \ProponeCompletion{} can be solved in polynomial time.
\end{restatable}
\begin{proof}
    Consider an instance $\langle N,M,\V,F\rangle$ of the completion problem where the frozen allocation is $F= ( F_1,F_2,\ldots,F_n)$. Let $U = M \setminus \bigcup_{i \in N}F_i$ be the set of unallocated goods. 

    First, observe that the maximin fair share $\mu_i$ for each agent $i\in N$ can be computed in polynomial time by iteratively assigning each unallocated good approved by agent $i$ to a bundle of minimum value from $i$'s viewpoint. The value $\min_{j \in N}v_i(F_j\cup C_j)$ of the worst bundle in the resulting allocation $(F_j \cup C_j)_{j \in N}$ is $i$'s maximin fair share $\mu_i$.

    Next, we show that given a target threshold $(\mu_i)_{i \in N}$, we can compute in polynomial time whether there exists an allocation that completes $F$ and gives value $\mu_i \in \mathbb{Z}$ for $i\in N$. To this end, we construct a flow network as follows: Create a source node $s$ and a sink node $t$. For each good $j \in U$, introduce a node $g_j$. For each agent $i\in N$, create a node $a_i$. For each good node $g_j$, create an arc $(s,g_j)$ with capacity $1$. For each agent node $a_i$, create an arc $(a_i,s)$ with unbounded capacity and a lower quota of $\mu_i$. For each good-agent pair $(g_j,a_i)$, if $v_{i}(j)=1$, then create an arc $(g_j,a_i)$ with capacity $1$. 

    A maximum flow in this network with the given edge constraints can be computed in polynomial time~(c.f. \citep{KleinbergTardos}). We know that a maximum flow in a network with integral capacity and lower quotas is always integral~(c.f.~\citep{KleinbergTardos}). Therefore, a flow meeting the constraints corresponds to a completion $C$ such that $(F_i \cup C_i)_{i \in N}$ is an \MMS{} allocation.

    Similarly, we can construct a network to compute a \Prop{1} completion as follows. For each agent $i\in N$, compute the required value to achieve the desired threshold as $\frac{1}{n}v_i(M) - v_i(F_i) -1$ and set this as the lower quota for each $(a_i,s)$ arc. 
    
    A flow meeting the above constraints corresponds to a completion $C$ such that $(F_i \cup C_i)_{i \in N}$ is \Prop{1}. To see this, consider the following two cases for any agent $i$. The first case is when agent $i$ approves at least one item among the bundles of other agents, namely $\cup_{j \neq i} A_j$, where $A_j = F_j \cup C_j$. In this case, any such approved good can be added to agent $i$'s bundle to achieve its proportionality threshold. The second case is when all goods in the bundles of other agents are of value $0$ to agent $i$. In this case, since agent $i$ receives all of its approved goods, its proportionality threshold is met and thus \Prop{1} is satisfied.
\end{proof}

Recall that for binary additive valuations, Pareto optimality is equivalent to assigning each good to an agent who approves it. Thus, to achieve \MMS{} and \PO{} (respectively, \Prop{1} and \PO{}), it suffices to check whether a given frozen allocation $F$ is Pareto optimal and whether we can obtain an \MMS{} allocation that completes $F$ by assigning goods to agents who approve them. Combining this observation with Theorem~\ref{thm:MMS:PROP1:binary}, we obtain the following corollary. 

\begin{restatable}{corollary}{BinaryPOMMSPROPone}\label{thm:PO:MMS:PROP1:binary}
For binary additive valuations, \MMSPOCompletion{} and \ProponePOCompletion{} can be solved in polynomial time.    
\end{restatable}

Although \Cref{thm:MMS:PROP1:binary} provides an efficient algorithm for checking the existence of a \Prop{1} or an \MMS{} completion, such a completion may not always exist; see Proposition~\ref{prop:binary:noalphaMMS} in Appendix~\ref{appendix:binary}.
Our next result identifies a sufficient condition for guaranteeing the existence of an \MMS{}+\PO{} completion.

\begin{restatable}{proposition}{MMSBinaryExistence}\label{prop:MMS:binary:existence}
For binary additive valuations, if the frozen allocation $F$ is \PO{} with respect to the frozen goods, then an \MMS{} and \PO{} allocation that completes $F$ always exists and can be computed in polynomial time. 
\end{restatable}
\begin{proof}
In the proof of Theorem~\ref{thm:MMS:PROP1:binary}, we have already seen that the maximin fair share $\mu_i$ for $i\in N$ can be computed in polynomial time for binary additive valuations.

To prove the claim, relabel the agents so that $\mu_1 \leq \mu_2 \leq \ldots \leq \mu_n$. 
For each agent $i\in N$, let $\Delta_{ik}$ denote the number of approved goods in the unallocated set $U$ that must be allocated to $k^\text{th}$ bundle in $i$'s \MMS{} partition, i.e., $\Delta_{ik} =\mu_i-v_i(F_k)$ each $k \in [n]$. 
Complete the allocation $F$ by letting each agent $i \in \{1,2,\ldots,n\}$ take, in the said order, exactly $\max \{\Delta_{ii},0\}$ unallocated goods for which it has value $1$. If there are still unallocated goods after the procedure, we allocate each of them to an agent who approves it if there is such an agent and allocate it arbitrarily otherwise. 

Clearly, after each agent $i$'s turn, $i$ obtains a bundle of value at least $\mu_i$. We claim that after each agent $i$'s turn, the maximin fair share of the remaining instance for each agent $j \in \{i+1,i+2,\ldots,n\}$ does not decrease, which implies that the algorithm allocates to $j$ a bundle of value at least $\mu_j$. To see this, consider any $j \geq i+1$. 
Observe that since $F$ is Pareto optimal with respect to the frozen goods, the goods in $F_i$ that are not approved by agent $i$ are also not approved by any other agent, which means that $v_j(F_i) \leq v_i(F_i)$. Thus, 
the number of approved goods that need to be allocated to $i^\text{th}$ bundle in $j$'s \MMS{} partition is at least that in $i$'s \MMS{} partition, namely, 
\[
\Delta_{ii} =\mu_i -v_i(F_i) \leq \mu_j - v_j(F_i)= \Delta_{ji}. 
\]
Thus, even after agent $i$ takes away $\Delta_{ii}$ goods, agent $j$ can complete the remaining frozen allocation $(F_{i+1},F_{i+2},\ldots,F_n)$ using $\Delta_{jk}$ approved goods for the $k^\text{th}$ bundle for $k \in \{i+1,i+2,\ldots,n\}$. Hence, the resulting allocation $A$ satisfies \MMS{}. Further, $A$ is \PO{} since each good is allocated to an agent who has a positive value for it. It is immediate that the algorithm runs in polynomial time. 
\end{proof}
\section{Lexicographic Valuations}\label{sec:lex}
In this section, we consider the class of lexicographic valuations. Recall that lexicographic valuations can be specified using only ordinal ranking over the goods. Furthermore, any allocation that is \EF{1}, Prop{1}, \MMS{}, or \PO{} with respect to this ordinal ranking also satisfies these properties with respect to any consistent cardinal realization.
\paragraph{Envy-Freeness Up to One Good.}
We show that the \EF{1} completion problem is \NPC{} even for lexicographic valuations. 

\begin{restatable}{theorem}{EFoneLexicographicNPhard}
\EFoneCompletion{} is \NPC{} even for lexicographic valuations.
\label{thm:EF1-Lexoicographic-NPc}
\end{restatable}

The detailed proof of Theorem~\ref{thm:EF1-Lexoicographic-NPc} is deferred to Appendix~\ref{appendix:lex} 
but a brief idea is as follows: The proof %
uses a reduction from \RainbowColoring{}, which is known to be \NPC{}~\citep{GJ79computers} and asks the following question: Given a hypergraph $H = (V,E)$ and a set of $k$ colors $\{c_1,c_2,\dots,c_k\}$, does there exist a coloring of the vertices in $V$ (i.e., an assignment of each vertex to one of $k$ colors) such that all vertices of each hyperedge have different colors? 

Observe that a necessary and sufficient condition for an allocation to be \EF{1} under lexicographic preferences is that no agent $j$ has two or more goods that another agent $i$ prefers over the favorite good in $i$'s bundle. The reduction in \Cref{thm:EF1-Lexoicographic-NPc} enforces this property by creating an agent for each hyperedge and setting up its preferences such that it prefers all `vertex goods' over the favorite frozen good in its bundle. Further, the construction ensures that all vertex goods are assigned only to the agents representing the $k$ colors. Thus, under any \EF{1} allocation, no two vertex goods are assigned to the same color agent, ensuring the rainbow property.

It is relevant to note that, unlike binary valuations where the value of a bundle depends only on the \emph{number} of approved goods, the preferences under lexicographic valuations depend on the \emph{identity} of the goods. This distinction between binary and lexicographic preferences also manifests in the proof techniques: To show the hardness of \EFoneCompletion{} under binary valuations~(\Cref{thm:EFoneBinaryAdditiveNPhard}), it was more natural to reduce from {\sc Equitable Coloring} where only the cardinality of the vertices needs to be balanced. By contrast, \EF{1} for lexicographic preferences requires the more preferred goods to go to separate agents, which is more naturally executed via \RainbowColoring{}.

The proof of Theorem~\ref{thm:EF1-Lexoicographic-NPc} uses ordinal valuations. Thus, if we do not require cardinal valuations in the definition of lexicographic preferences, \EFoneCompletion{} becomes strongly NP-hard. However, if we represent the valuations as numbers (for example, powers of 2), it is not clear whether strong NP-hardness holds for lexicographic valuations.

\paragraph{Proportionality Up to One Good and Maximin Share.}
Although checking the existence of an \EF{1} completion is \NPC{}, there is an efficient algorithm for checking whether there exists a completion satisfying the weaker fairness notion of \Prop{1}~(\Cref{prop:LexProp1}).

\begin{proposition}
Under lexicographic valuations, every allocation satisfies \Prop{1}. Thus, \ProponeCompletion{} can be solved in polynomial time.
\label{prop:LexProp1}
\end{proposition}

The reasoning behind \Cref{prop:LexProp1} is that under lexicographic valuations, the favorite good of an agent is more valuable than all the other goods combined. Thus, under any allocation, either an agent receives its favorite good, or, if it doesn't, then it can hypothetically add this good to its bundle. 

Under lexicographic valuations, an allocation is Pareto optimal if and only if it satisfies \emph{sequencibility}~\citep{HSV+21fair}.\footnote{An allocation is said to be \emph{sequencible} if there exists a picking sequence that induces it. Formally, an allocation $A$ is sequencible if there exists a string $\langle a_1, a_2, \dots, a_m \rangle$, where each $a_i \in N$ denotes the index of an agent, such that starting from $a_1$, if each agent $a_i$ picks its favorite remaining item on its turn, then the resulting allocation is $A$.} Therefore, the problem of finding a \PO{} allocation that completes a given frozen allocation reduces to finding a completion such that the final allocation has a corresponding picking sequence of agents. The latter problem can be solved efficiently by constructing a picking sequence for the frozen allocation and extending it for the unallocated goods~(see Theorem \ref{thm:PO:lexiocographic} in Appendix~\ref{appendix:lex}). 
The following is an immediate consequence of this observation and Proposition~\ref{prop:LexProp1}.

\begin{theorem}\label{thm:PROP1-PO-Lex-Poly}
For lexicographic valuations, \ProponePOCompletion{} can be solved in polynomial time.
\end{theorem}

In the standard fair division setting, Hosseini et al.~\cite{HSV+20fairAAAI} obtain a characterization of an \MMS{} allocation for lexicographic valuations. They show that an allocation is \MMS{} if and only if each agent obtains one or more of their top $(n-1)$ goods or all of their bottom $(m-n+1)$ goods. Such an allocation exists and can be computed in polynomial time. 
By contrast, an \MMS{} completion can fail to exist for lexicographic valuations even when the frozen allocation is Pareto optimal.%

\begin{restatable}{proposition}{NoMMSlexicographic}\label{prop:noMMS:lexicographic}
There exists an instance with lexicographic valuations for which an \MMS{} allocation that completes the given frozen allocation $F$ may not exist even when $F$ is Pareto optimal with respect to the frozen goods. 
\end{restatable}
\begin{proof}
Consider an instance with four goods $f_1,f_2,g_1,g_2$ and two agents with the following ranking over the goods:
\begin{align*}
\text{agent }1 & : g_1 \>_1 g_2 \>_1 \underline{f_1} \>_1 f_2\\
\text{agent }2 & : f_1 \>_2 \underline{f_2} \>_2 g_1 \>_2 g_2
\end{align*}
The underlined entries indicate the frozen good assigned to the agent. 
Note that the frozen allocation $F$ is Pareto optimal with respect to the frozen goods. The \MMS{} partition of agent $1$ is $(F_1 \cup \{g_2\}, F_2 \cup \{g_1\})$ and hence its maximin fair share is $v_1(F_1 \cup \{g_2\})$. On the other hand, the \MMS{} partition of agent $2$ is $(F_1, F_2 \cup \{g_1,g_2\})$ and hence its maximin fair share is $v_2(F_2 \cup \{g_1,g_2\})$. Thus, to achieve \MMS{}, we need to allocate at least one of $g_1$ and $g_2$ to agent $1$, and, simultaneously, both of $g_1$ and $g_2$ to agent $2$, which is impossible.
\end{proof}

Nevertheless, we can design an efficient algorithm to compute the maximin share for each agent for lexicographic valuations. Intuitively, it suffices to identify frozen bundles, say $F_1,\ldots,F_k$ with $v_i(F_1) \leq \ldots \leq v_i(F_k)$, whose values are less than the value of the most valuable unallocated good and allocate each of the top $(k-1)$ unallocated goods to the first $k-1$ bundles and the remaining goods to the last bundle. 

\begin{restatable}{proposition}{MMSvaluelexicographic}\label{prop:computeMMS:lexicographic}
For lexicographic valuations, the maximin share value can be computed in polynomial time.   
\end{restatable}

The detailed proof of the above proposition is deferred to Appendix~\ref{appendix:lex}. 
Proposition~\ref{prop:computeMMS:lexicographic} implies that \MMSCompletion{} belongs to NP when agents have lexicographic valuations. Moreover, by exploiting the structure of an \MMS{} partition, we establish the polynomial-time solvability of \MMSCompletion{}. A crucial property here is that to ensure \MMS{} for each agent, it is sufficient to allocate either the empty bundle, a single good, or a segment of their bottom-preferred unallocated goods. This implies that if an \MMS{} allocation exists, there is one with at most one agent of the last type.

\begin{restatable}{theorem}{LexicographicMMS}\label{thm:decideMMS:lexicographic}
For lexicographic valuations, \MMSCompletion{} can be solved in polynomial time.    
\end{restatable}

\section{General Additive Valuations}\label{sec:additive}

In this section, we consider the class of general additive valuations.

\paragraph{Envy-Freeness Up to One Good and Proportionality Up to One Good.}

Our first two results in this section show that deciding the existence of an \EF1{} completion is computationally intractable even for \emph{two} agents with additive valuations and for three agents with \emph{identical} additive valuations. Given the intractability of \EFoneCompletion{}, a natural approach in search of tractability results is to weaken the fairness notion from \EF{1} to \Prop{1}. Unfortunately, using similar proof techniques, it can be shown that the completion problem remains computationally hard even for \Prop{1}.

\begin{restatable}{theorem}{EFoneTwoAgentAdditiveNPhard}
\EFoneCompletion{} and \ProponeCompletion{} are \NPC{} even for two agents with additive valuations.
\label{thm:EFoneTwoAgentAdditiveNPhard}
\end{restatable}

\begin{restatable}{theorem}{EFoneThreeAgentIdenticalAdditiveNPhard}
\EFoneCompletion{} and \ProponeCompletion{} are \NPC{} even for three agents with identical additive valuations.
\label{thm:EFoneThreeAgentIdenticalAdditiveNPhard}
\end{restatable}

Observe that under identical valuations, every allocation is Pareto optimal. Thus, \Cref{thm:EFoneThreeAgentIdenticalAdditiveNPhard} also implies that checking the existence of an \EF{1}+\PO{} (and, similarly, \Prop{1}+\PO{}) completion is \NPH{}.

\begin{restatable}{corollary}{EFonePOThreeAgentAdditiveNPhard}\label{cor:EFonePOThreeAgentAdditiveNPhard}
\EFonePOCompletion{} and \ProponePOCompletion{} are \NPH{} even for three agents with identical additive valuations.
\end{restatable}

\Cref{thm:EFoneTwoAgentAdditiveNPhard,thm:EFoneThreeAgentIdenticalAdditiveNPhard} do not impose any structural restrictions on a frozen allocation. Our next result (\Cref{thm:ef1_acyclic}) shows that if the frozen allocation is ``fair'', we can circumvent the intractability of the completion problem. Formally, given a partial allocation $A$, let us define its \emph{envy graph} $G_A$ as a directed graph whose vertices are the agents and there is an edge $(i,j)$ if and only if agent $i$ envies agent $j$ under $A$, i.e., $v_i(A_i) < v_i(A_j)$~\citep{LMM+04approximately}. \Cref{thm:ef1_acyclic} states that if the frozen allocation is \EF{1} and its envy graph is acyclic, then an \EF{1} completion always exists and can be efficiently computed.

\begin{restatable}{proposition}{EFoneenvygraph}\label{thm:ef1_acyclic}
Suppose that a frozen allocation $F$ is \EF{1} and the envy graph with respect to $F$ is acyclic. Then, an \EF{1} allocation that completes $F$ always exists and can be computed in polynomial time. 
\end{restatable}
\begin{proof}
Consider the topological ordering induced by the envy graph $G_F$ of the frozen allocation $F$. Assume, without loss of generality, that the resulting ordering is $\sigma = (1,2,\ldots,n)$. Thus, for any pair of agents $i,j \in N$ with $i > j$, $i$ does not envy $j$ under the allocation $F$, i.e., $v_i(F_i) \geq v_i(F_j)$. 

To assign the unallocated goods $U=M\setminus \cup_{i \in N}F_i$, we apply the round-robin procedure with respect to the ordering $\sigma$ (i.e., first agent $1$ picks, then agent $2$, and so on). At its turn, each agent $i \in \{1,2,\ldots,n\}$ picks its favorite good in $U$ that has not been chosen. 

To show that the resulting allocation $A=(F_i \cup C_i)_{i \in N}$ is \EF{1}, consider any pair of distinct agents $i,j \in N$. Suppose $i > j$. Then, $i$ does not envy $j$ under $F$, i.e., $v_i(F_i) \geq v_i(F_j)$, and $i$'s envy toward $j$ is bounded up to one good in the completion phase, i.e., $C_j$ is empty or $v_i(C_i) \geq v_i(C_j \setminus \{g\}))$ for some $g \in C_j$. Thus, by additivity of valuations, we have that either $v_i(A_i) \geq v_i(A_j)$ or $v_i(A_i) \geq v_i(A_j \setminus \{g\}))$ for some $g \in A_j$. Now suppose $i < j$. Then, $i$ does not envy $j$ under the completion phase, i.e., $v_i(C_i) \geq v_i(C_j)$, and $i$'s envy toward $j$ is bounded up to one good in $F$, i.e., $F_j$ is empty or $v_i(F_i) \geq v_i(F_j \setminus \{g\}))$ for some $g \in F_j$. In either case, $i$'s envy toward $j$ is bounded up to one good under $A$.
\end{proof}

Note that the two assumptions in Proposition~\ref{thm:ef1_acyclic} are necessary to guarantee the existence of an EF1 completion and its polynomial-time computability.
In fact, if the frozen allocation is \EF{1} but the envy graph contains a cycle, the completion problem becomes computationally hard. For example, we can construct an instance where the frozen allocation is \EF{1}, but every agent envies another (e.g., each agent receives an item valued at 0, while all other agents receive items valued at a very large amount from their perspective). This setup essentially reduces the problem to deciding whether there exists an envy-free allocation of the unallocated items, which is known to be NP-hard even for binary additive valuations~\cite{BouveretL08}.
Moreover, if the envy graph of the frozen allocation is acyclic, our \Cref{thm:EFoneThreeAgentIdenticalAdditiveNPhard} already demonstrates that the completion problem remains hard (recall that for identical valuations, any allocation is Pareto-optimal).

Building on Proposition~\ref{thm:ef1_acyclic}, we can show that for two agents with identical valuations, there is a polynomial-time algorithm for checking the existence of an \EF{1} or \Prop{1} completion. Note that \EF{1} implies \Prop{1} but the opposite may not hold. Indeed, an allocation that assigns two goods to agent $1$ and nothing to agent $2$ is \Prop{1} but not \EF{1}. Nevertheless, a technique similar to that for \EF{1} also works for \Prop{1} and two agents with identical additive valuations. 

\begin{restatable}{theorem}{EFonetwo}\label{thm:ef1_two_identical}
For two agents with identical additive valuations, \EFoneCompletion{} and \ProponeCompletion{}  can be solved in polynomial time.
\end{restatable}

As shown in \Cref{thm:EFoneTwoAgentAdditiveNPhard,thm:EFoneThreeAgentIdenticalAdditiveNPhard}, the \EFoneCompletion{} problem is \NPC{} for three agents with identical valuations or two agents with possibly non-identical valuations. Thus, our results present a tight demarcation of the frontier of tractability for the \EFoneCompletion{} problem. Further, we note that \EFoneCompletion{} problem for a general number of agents is \emph{strongly} \NPC{}; this can be observed by carrying out a reduction from \textsc{3-Partition}. This result is presented as~\Cref{thm:EFone_strongly_NPC} in~\Cref{appendix:additive}.
\paragraph{Maximin Share.}
For additive valuations, an \MMS{} allocation can fail to exist~\citep{KP18fair}, but a $(3/4 + 3/3836)$-\MMS{} allocation always exists~\citep{AG24breaking}. By contrast, in the completion problem, it is impossible to guarantee $\alpha$-\MMS{} for any $\alpha \in (0,1]$. 
Intuitively, the failure occurs for an instance where each agent $i$ thinks that the first $i$ agents should equally divide the unallocated goods and the remaining agents are already rich enough. This discrepancy in agents' valuations prohibits the existence of an $\alpha$-\MMS{} completion. 

\begin{restatable}{proposition}{NoalphaMMS}\label{prop:noalphaMMS}
For every $\alpha \in (0,1]$, there is an instance for which no $\alpha$-\MMS{} allocation that completes a given frozen allocation $F$ exists even when the frozen allocation $F$ is Pareto optimal with respect to the frozen goods and each bundle in $F$ contains at most one good. 
\end{restatable}

\begin{proof}
For every $n \in \mathbb{N}$, let $H_n \coloneqq \sum^n_{j=1}\frac{1}{j}$ denote the $n^\textup{th}$ harmonic number.
Consider any $\alpha \in (0,1]$, and let $n \in \mathbb{N}$ be such that $1 < \frac{\alpha}{2}H_n$; such $n$ must exist since $H_n$ is strictly monotone increasing. Let $0 < \epsilon_1 < \epsilon_2 < \ldots < \epsilon_n < 1$. 

Consider an instance of the completion problem with $n$ agents, $n$ frozen goods, and $\ell$ unallocated goods with $\ell \geq n$.    
Construct a frozen allocation $F=(F_1,\ldots,F_n)$ where for each agent $i \in N$, $F_i$ consists of a single frozen good $f_i$ such that %
$v_i(f_j) = \epsilon_j$ if $j \leq i$ and  $v_i(f_j) = \frac{\ell}{i}$ if $j > i$. Here, we have that for $i \in N$, $\ell \geq i$ and each agent $i \in [n]$ has the same preference ranking over the bundles, i.e., $v_i(F_1)  \leq v_i(F_2) \leq \ldots \leq v_i(F_n)$. Further, $F$ is Pareto optimal with respect to the frozen goods since every agent $i \in N$ strictly prefers good $f_i$ over $f_{i-1},f_{i-2},\ldots,f_{1}$. 

Suppose every agent has a value of $1$ for every unallocated good. Note that the \MMS{} value $\mu_i$ of agent $i \in N$ is at least $\lfloor \frac{\ell}{i} \rfloor$; this can be achieved by assigning unallocated goods among the first $i$ bundles as equally as possible.  

Suppose towards a contradiction that there exists an $\alpha$-\MMS{} allocation that completes $F$. Then, agent $i \in N$ must be assigned at least $\alpha \cdot \lfloor \frac{\ell}{i} \rfloor \geq \alpha \cdot  \frac{\ell}{2i}$ unallocated goods. Since there are $\ell$ unallocated goods, it must be that $\ell \geq \sum_{i \in N}\alpha \cdot  \frac{\ell}{2i}$, or, equivalently, $1 \geq \frac{\alpha}{2} H_n$. But this contradicts the assumption that $1  < \frac{\alpha}{2}H_n$.
\end{proof}

In the standard setting of fair division, the decision problem of determining the existence of an \MMS{} allocation is in $\Sigma^2_p$~\citep{BL16characterizing}; however, it remains open whether the problem is \NPH{} or even harder. A simple reduction from {\sc PARTITION} shows that the search version of the problem is \NPH{}. It is easy to see that \MMSCompletion{} includes the above decision problem as a special case. Our final result shows that \MMSCompletion{} is \NPH{} even for two agents with additive valuations. 

\begin{restatable}{theorem}{MMSTwoAgentAdditiveNPhard}
\MMSCompletion{} is \NPH{} even for two agents with additive valuations.
\label{thm:MMSTwoAgentAdditiveNPhard}
\end{restatable}

\begin{proof}
Given an instance of \textsc{PARTITION}, a set of positive integers $\{w_1,w_2,\cdots,w_m\}$ such that $\sum_{i\in [m]} w_i = 2T$, we create an instance of the completion problem as follows. 
There are two agents $1$ and $2$. Let $0<\epsilon<1$.
There are two frozen goods $f_1, f_2$ where agent $1$ receives $f_1$ and agent $2$ receives $f_2$. Agent $1$ has value $0$ for $f_1$ and $\epsilon$ for $f_2$. Agent $2$ has value $0$ for $f_2$ and $\epsilon$ for $f_1$. 
For each $k \in [m]$, we create a good $k$ that each agent values at $w_k$. Denote by $(F_i)_{i \in \{1,2\}}$ the frozen allocation.

We show that there exists an \MMS{} allocation that completes the frozen allocation if and only if the original \textsc{PARTITION} problem has a yes-instance, i.e., a partition of the unallocated goods into two bundles of value $T$. 

First, suppose that there exists a \MMS{} allocation $A = (F_i \cup C_i)_{i =1,2}$ that completes the frozen allocation $F$, where $(C_i)_{i=1,2}$ is a completion. Let $\mu^*=\max_{(S_1,S_2) \in \Pi_U}  \min \{\sum_{j \in S_1} w_j, \sum_{j \in S_2} w_j\}$. Suppose towards a contradiction that \textsc{PARTITION} problem is a no-instance, namely $\mu^* < T $. Then observe that for each agent $i \in N$, 
$$
\mu_i \geq \epsilon + \mu^*. 
$$
This means that $v_i(C_i) = v_i(F_i \cup C_i) \geq \epsilon + \mu^*$, yielding
\[
\min \left\{\sum_{j \in C_1} w_j, \sum_{j \in C_2} w_j\right\} = \min \{v_1(C_1),v_2(C_2)\} > \mu^*,
\]
contradicting the definition of $\mu^*$. Thus, \textsc{PARTITION} problem is a yes-instance. 

Conversely, suppose that there exists a partition $(S,[m] \setminus S)$ of $[m]$ such that $\sum_{i \in S} w_i = \sum_{i \in [m] \setminus S} w_i = T$. Then, the maximin fair share $\mu_i$ of each agent $i$ is $T$. Thus, allocating bundle $S$ to agent $1$ and $[m] \setminus S$ to agent $2$ yields an \MMS{} allocation. 
\end{proof}
\section{Conclusion}\label{sec:conclusion}
We initiated the study of fair and efficient completion of indivisible goods and provided several algorithmic and intractability results for this problem under additive valuations as well as its subclasses. Going forward, it would be interesting to study other types of resources such as \emph{chores}, combination of goods and chores, and combination of divisible and indivisible resources~\citep{BLL+21fair}. Additionally, while our work only considers the case of \emph{frozen} goods, it would be interesting to also consider \emph{forbidden} assignments wherein certain items cannot be assigned to certain agents under any completion.

In our paper, we defined the notion of \MMS{}, which incorporates a pre-determined frozen allocation. An alternative approach would be to use the standard \MMS{} definition that disregards the frozen allocation and study the problem of achieving this \MMS{}. Notably, with the standard \MMS{} definition, our computational complexity results remain unchanged. For binary additive valuations, it is still possible to compute an \MMS{} threshold in polynomial time, allowing us to apply the same network-flow argument to compute an \MMS{} allocation. For lexicographic valuations, under the \MMS{} definition, each agent's \MMS{} partition should take the form $(\{g_1\},\{g_2\}, \{g_{n},\ldots,g_m\})$~\cite{HSV+20fairAAAI} where $g_1$ is the most favorite good, and $g_2$ is the second favorite good, and so on. After removing agents whose frozen bundles already achieve their \MMS{} threshold, one can iteratively guess which agent should receive a bundle of $m-n+1$ items, create a bipartite graph where there is an edge between an unallocated item $g$ and an agent $i$ such that $i$ has an \MMS{} value for $g$, and decide whether there exists a matching that saturates the agents in the graph. For general additive valuations, the NP-hardness of deciding the existence of an \MMS{} allocation (Theorem~\ref{thm:MMSTwoAgentAdditiveNPhard}) still holds. This is because each agent assigns a very small value to the frozen items, resulting in the same thresholds for both our \MMS{} definition and the standard \MMS{} definition.
\section*{Acknowledgments}

We thank anonymous reviewers 
for helpful comments. Ayumi Igarashi acknowledges support from JST FOREST under grant number JPMJFR226O. Rohit Vaish acknowledges support
from DST INSPIRE grant no. DST/INSPIRE/04/2020/000107 and
SERB grant no. CRG/2022/002621. Vishwa Prakash HV acknowledges support of TCS Research Scholar Fellowship. 

\clearpage
\bibliographystyle{abbrv}
\bibliography{references}

\appendix
\newpage

\begin{center}
    \Large{Appendix}    
\end{center}

\section{Omitted Material from Section~\ref{sec:binary}}\label{appendix:binary}

\begin{proposition}\label{prop:MNW}
There exists an instance with three agents with binary additive valuations in which no MNW completion of the given frozen allocation $F$ satisfies \EF{1}, even though an \EF{1}+\PO{} completion exists.
\end{proposition}
\begin{proof}
Consider an instance where there are eight goods and three agents with the binary additive valuations, specified in Table~\ref{table:MNW}. 

\begin{table}[htb]
	\footnotesize
	\centering
	\begin{tabular}{lccccccccc}
		\toprule
        \multicolumn{1}{l}{Goods} & $g_1$ & $g_2$ & $g_3$ & $g_4$ & $f_1$ & $f_2$ & $f_3$ & $f_4$\\ 
        \midrule
        Agent 1 & $\boxed{1}$ & $1$ & $0$ &  $0$ & $0$ & $0$ & $0$ &  $0$ \\  
        Agent 2 & $1$ & $\boxed{1}$ & $\boxed{1}$ &  $\boxed{1}$ & $1$ & $1$ & $1$ &  $1$\\  
        Agent 3 & $0$ & $0$ & $0$ &  $0$ & 
        \underline{$1$}&  \underline{$1$} &  \underline{$1$}&   \underline{$1$}\\  
		\bottomrule
	\end{tabular}
 \vspace{0.1in}
\caption{An instance for which an MNW completion may not satsify \EF{1}. The underlined good indicates the frozen good assigned to the agent. The boxed goods indicate the \PO{} and \EF{1} completion.}
\label{table:MNW}
\end{table}

There are four frozen goods $f_1,\dots,f_4$ and four unallocated goods $g_1,\dots,g_4$. In the frozen allocation $F$, all frozen goods are given to agent $3$. Any MNW completion must allocate the goods $g_3$ and $g_4$ to agent $2$ and the goods $g_1$ and $g_2$ to agent 1. This allocation $A$ is not \EF{1} since agent $2$ envies agent $3$ even after removing any single good from agent $3$'s bundle. However, there is an \EF{1}+\PO{} allocation that completes $F$, namely one that assigns $g_1$ to agent 1 and $g_2,g_3,g_4$ to agent $2$. 
\end{proof}

\begin{proposition}\label{prop:binary:noalphaMMS}
For each $\alpha \in (0,1]$, there exists an instance with binary additive valuations where no $\alpha$-MMS allocation that completes the given frozen allocation exists. 
\end{proposition}
\begin{proof}
Take any $\alpha \in (0,1]$. Let $x,y \in \mathbb{Z}_+$ such that $\alpha \geq \frac{x}{y}$. 
Consider an instance where there are $n > \frac{y}{x}$ agents and $(n+1)y$ goods. Agents have binary additive valuations specified in Table~\ref{table:binary:alphaMMS}, where each agent $i$ has value $0$ for $g_{i1},\ldots,g_{iy}$ and value $1$ for the remaining goods. 

\begin{table*}[ht]
	\footnotesize
	\centering
	\begin{tabular}{lccccccccccccccccc}
		\toprule
        \multicolumn{1}{l}{Goods} & $g_1$ & $\cdots$ & $g_y$ & $g_{11}$ & $\cdots$ & $g_{1y}$ & $g_{21}$ & $\cdots$ & $g_{2y}$ & $\cdots$ & $g_{n1}$ & $\cdots$ & $g_{ny}$ \\ 
        \midrule
        Agent 1 & $1$ & $\cdots$ & $1$ & \underline{$0$} & $\cdots$ & \underline{$0$} & $1$ & $\cdots$ & $1$ & $\cdots$ & $1$ & $\cdots$ & $1$\\  
        Agent 2 & $1$ & $\cdots$ & $1$ & $1$ & $\cdots$ & $1$ & \underline{$0$}& $\cdots$ & \underline{$0$}& $\cdots$ & $1$ & $\cdots$ & $1$\\  
        ~~~~$\vdots$& $\vdots$ & $\vdots$ & $\vdots$ &$\vdots$ & $\vdots$ & $\vdots$ & $\vdots$ &$\vdots$ & $\vdots$ & $\vdots$ & $\vdots$ &$\vdots$ & $\vdots$ \\ 
        Agent $n$ & $1$ & $\cdots$ & $1$ & $1$ & $\cdots$ & $1$ & $1$ & $\cdots$ & $1$ & $\cdots$ & \underline{$0$} & $\cdots$ & \underline{$0$}\\ 
		\bottomrule
	\end{tabular}
\caption{An instance for which no $\alpha$-MMS allocation exists for any $\alpha>0$. The underlined good indicates the frozen good assigned to the agent.}
\label{table:binary:alphaMMS}
\end{table*}

In a frozen allocation $F$, each agent $i$ receives $g_{i1},\ldots,g_{iy}$. The maximin fair share for each agent is $y$. Thus, any $\alpha$-MMS allocation that completes the frozen allocation should allocate at least $x$ goods from $g_1,\ldots,g_y$ to each agent. Such an allocation does not exist since $n \cdot x > y$.~\qedhere
\end{proof}

\section{Omitted Material from Section~\ref{sec:lex}}\label{appendix:lex}

\EFoneLexicographicNPhard*

\begin{proof}
    Checking whether a given allocation correctly completes a given frozen allocation and whether it satisfies \EF{1} can be done in polynomial time. Thus, the problem is clearly in \NP{}.

    To prove \NP{}-hardness, we will show a reduction from \RainbowColoring{}, which is known to be \NPC{}~\citep{GJ79computers} and asks the following question: Given a hypergraph $H = (V,E)$ and a set of $k$ colors $\{c_1,c_2,\dots,c_k\}$, does there exist a coloring of the vertices in $V$ (i.e., an assignment of each vertex to one of the $k$ colors) such that all vertices of any fixed hyperedge have different colors?\\

    \emph{Construction of the reduced instance}. Let $q \coloneqq |V|$ and $r \coloneqq |E|$ denote the number of vertices and hyperedges, respectively. We will assume that for every vertex $v \in V$, 
    there are at least two singleton hyperedges in $E$ containing only that vertex, say $e' \coloneqq \{v\}$ and $e'' \coloneqq \{v\}$. Note that adding such singleton hyperedges to the \RainbowColoring{} instance does not affect the YES/NO answer.

    Given any instance of \RainbowColoring{}, we will construct an instance of the completion problem as follows: There are three sets of agents, namely $r$ \emph{hyperedge} agents $\{a_e\}_{e \in E}$, $k$ \emph{color} agents $\{c_1,\dots,c_k\}$, and $q \cdot r$ \emph{pair} agents $\{p_{e,v}\}_{e \in E, v \in V}$. Thus, there are $r+k+qr$ agents overall.
    
    The set of goods consists of $q$ \emph{vertex} goods $\{v_1,\dots,v_q\}$, one for each vertex in $V$, and $r+k+qr$ \emph{frozen} goods, one for each agent. The frozen goods comprise of frozen hyperedge goods $\{g_e\}_{e \in E}$, frozen color goods $\{g_{c_1},\dots,g_{c_k}\}$, and frozen pair goods $\{g_{e,v}\}_{e \in E, v \in V}$. Each frozen good is assigned to the agent after which it is named. The vertex goods constitute the set of unallocated goods.

    The lexicographic preferences of the agents are specified as follows: 
    \begin{itemize}
        \item For each hyperedge $e \in E$, if $e = \{v_{j_1},\dots,v_{j_\ell}\}$, then the preference of the hyperedge agent $a_e$ is given by
    \begin{align*}
        &a_e: \{ v_{j_1}, \dots, v_{j_\ell} \} \> \\
        &\langle \text{all frozen hyperedge and all frozen pair goods except \(g_e\)} \rangle \\
        &\> \underline{g_e} \> \langle \text{rest} \rangle.
    \end{align*}
    The ordering within the vertex goods $v_{j_1}, \dots, v_{j_\ell}$ or within the frozen hyperedge and frozen pair goods can be arbitrary. The set $\langle \text{rest} \rangle$ denotes the remaining goods, and the ordering of these goods can also be arbitrary. The underlined good indicates the frozen good assigned to the agent.
        \item For every $i \in [k]$, the preference of color agent $c_i$ is given by
    \begin{align*}
        c_i: \underline{g_{c_i}} \> \langle \text{rest} \rangle.
    \end{align*}
        \item For every hyperedge $e \in E$ and every vertex $v \in V$ (not necessarily in $e$), the preference of the pair agent $p_{e,v}$ is given by
    \begin{align*}
        p_{e,v}: v \> g_e & \> \langle \text{all frozen pair goods except for } g_{e,v} \rangle   \\
        &\> \underline{g_{e,v}} \> \langle \text{rest} \rangle.
    \end{align*}
    \end{itemize}

    Let us now show that the \RainbowColoring{} instance admits a solution if and only if the reduced instance of completion problem does.\\

    Proof of \RainbowColoring{} $\Rightarrow$ \EF{1} Completion.\\ 
    
    Let $\sigma: V \rightarrow [k]$ be a valid rainbow coloring of the hypergraph $H = (V,E)$. Consider an assignment of the unallocated goods such that, for every vertex $v$, the corresponding vertex good is assigned to the color agent $c_{\sigma(v)}$.

    Observe that the aforementioned allocation, say $A$, is a valid completion of the given frozen allocation. Furthermore, the allocation $A$ satisfies \EF{1} for the following reason: Each color agent $c_i$ receives its top-ranked good and therefore, due to lexicographic preferences, does not envy any other agent. Each pair agent $p_{e,v}$ receives only the frozen good $g_{e,v}$, but it does not \EF{1}-envy any other agent because all of its more preferred goods are assigned to \emph{different} agents; indeed, its favorite good $v$ is assigned to a color agent, its second-favorite good $g_e$ is assigned to a hyperedge agent, and the frozen pair goods other than $g_{e,v}$ are assigned to distinct pair agents. Finally, for each hyperedge agent $a_e$, observe that by the rainbow property, the vertex goods ranked at the top of agent $a_e$'s preference list are assigned to distinct color agents. Furthermore, in $a_e$'s preference list, all frozen color goods are ranked below the frozen hyperedge good $g_e$ assigned to $a_e$. Thus, the envy from agent $a_e$ towards any color agent is bounded by the removal of a vertex good. A similar argument can be used to bound the envy from $a_e$ towards any other hyperedge or pair agent. Therefore, the allocation $A$ satisfies \EF{1}.\\

    Proof of \EF{1} Completion $\Rightarrow$ \RainbowColoring{}.\\ 
    
    Suppose allocation $A$ is a valid \EF{1} completion of the given frozen allocation. We will first show that all vertex goods must be assigned among the color agents in $A$. Then, we will argue that the coloring induced by $A$ in the hypergraph $H$ has the desired rainbow property.

    To see why each vertex good is assigned to a color agent, consider the preference list of any pair agent $p_{e,v}$. This agent owns the frozen pair good $g_{e,v}$, which is ranked below all other frozen pair goods (which are assigned to other pair agents). Recall that for any vertex $v$, there exist at least two pair agents of the form $p_{e',v}$ and $p_{e'',v}$ that both rank $v$ as their favorite good. If the vertex good $v$ is assigned to some pair agent, then \EF{1} is violated for at least one of the pair agents $p_{e',v}$ and $p_{e'',v}$. Thus, a vertex good cannot be assigned to any pair agent.

    By a similar reasoning, it can be seen that a vertex good cannot be assigned to any hyperedge agent. Indeed, for any vertex $v$ and hyperedge $e$, the pair agent $p_{e,v}$ prefers the frozen good $g_e$ assigned to the hyperedge agent $a_e$ over its own frozen good $g_{e,v}$. Therefore, in order to satisfy \EF{1} from the perspective of the pair agent $p_{e,v}$, the vertex good $v$ cannot be assigned to the hyperedge agent $a_e$. Thus, each vertex good must be assigned to a color agent. One can now naturally infer a coloring of the hypergraph $H = (V,E)$. It remains to be shown that this coloring has the rainbow property.

    Consider any hyperedge $e \in E$ and the corresponding hyperedge agent $a_e$. We have argued above that any vertex good must be assigned to a color agent. If two or more vertex goods corresponding to the vertices of hyperedge $e$ are assigned to the same color agent, then \EF{1} is violated from $a_e$'s perspective. Therefore, all vertices in hyperedge $e$ must have distinct colors, implying a rainbow coloring.
\end{proof}

\begin{theorem}\label{thm:PO:lexiocographic}
For lexicographic valuations, \POCompletion{} can be solved in polynomial time.
\end{theorem}

To prove Theorem~\ref{thm:PO:lexiocographic}, we present two lemmas. 
Given a partial allocation \(A\), where \( M' = \cup_{i=1}^n A_i \subseteq M\), we define \(I^A = \langle N,M',\mathcal{V} \rangle\) to be the \emph{sub-instance} of the fair allocation problem over the goods in \(M'\) and agent set \(N\). Given a frozen allocation \(F\), the following lemma gives us a simple check to test whether there is a \PO{} allocation that completes the given frozen allocation \(F\). Formally, a partial allocation $A$ is \emph{sequencible} if there exists a picking sequence $\sigma=(a_1,a_2,\ldots,a_m)$ of the agents so that $A$ can be realized by letting each agent $a_j$ take the favorite good among the remaining goods for $j \in \{1,2,\ldots,m\}$. 

\begin{lemma}
    If the frozen allocation $F$ is not sequencible in $I^F$, then there exists no sequencible allocation that completes $F$.
\end{lemma}
\begin{proof}
We prove the contrapositive. 
    Suppose that there is a complete allocation $A$  that is sequencible and completes $F$. Consider the $m$-length picking sequence $\sigma = (a_1,a_2,\ldots,a_m )$ that realizes $A$. 
    Let $U=M\setminus M'$. For each good $g \in U$, delete the corresponding entry of the agent who picks $g$ in $\sigma$. When this procedure terminates, the remainder sequence corresponds to a picking sequence for the instance \(I^F\) on frozen goods.
\end{proof}

Now, given a frozen allocation \(F\), if \(F\) is sequencible in \(I^F\), then Algorithm~\ref{alg:GetSEQ} gives a picking sequence that yields $F$. Otherwise, it detects that \(F\) is not sequencible. 

\begin{algorithm}[t]
    \caption{GetSEQ($F$) }
    \label{alg:GetSEQ}
    \begin{algorithmic}[1]
        \If{ No agent in $I^F$ gets their top good}
            \State \textbf{return} ``NOT SEQ''
        \ElsIf {$I^F$ has no good}
        \State \textbf{return} $()$ \Comment Empty list
            \EndIf
        \State Let $a$ be the agent who gets her top good $g$ in \(F\)
        \State \textbf{return} $(a, \text{GetSEQ}(F \setminus \{g\}) )$
    \end{algorithmic}
\end{algorithm}

If the given frozen allocation is sequencible, we can allocate the remaining goods iteratively, ensuring that the sequencibility of the partial allocation remains an invariant throughout the process.

\begin{lemma}
    If the frozen allocation \(F\) is sequencible in \(I^F\), then a sequencible allocation that completes $F$ exists and can be computed in polynomial time. 
    \label{lem:extend:frozen}
\end{lemma}
\begin{proof}

To complete the allocation \(F\), we iteratively allocate each unallocated good in \(U\) and extend a corresponding picking sequence. 

Let $M' = \cup_{i=1}^n F_i$. First, we compute a corresponding picking sequence \(\sigma=(a'_1,a'_2,\ldots,a'_{|M'|})\) on the frozen goods \(M'\subseteq M\) by Algorithm~\ref{alg:GetSEQ}.  

For each unallocated good \(g^*\notin M'\), we run the picking sequence \(\sigma\) on \(M'\cup\{g^*\}\). Consider the following two cases. 
\begin{itemize}
\item If some agent \(a'_i\) picks \(g^*\), we allocate \(g^*\) to \(a'_i\) and we extend \(\sigma\) by giving \(a'_i\) another chance to pick the good, namely, update $\sigma=(a'_1,a'_2,\ldots,a'_i,a'_i,\ldots,a'_{|M'|})$. 
\item If no agent picks \(g^*\), update \(\sigma=(\sigma,a'_1)\). 
\end{itemize} 
It is easy to see that sequencibility is maintained as an invariant while allocating each unallocated good from \(U\), to complete the allocation \(F\).
\end{proof}

\begin{proof}[Proof of Theorem~\ref{thm:PO:lexiocographic}]
Given a frozen allocation \(F\), we first check if \(F\) is sequencible in \(I^F\) using Algorithm~\ref{alg:GetSEQ}, which runs in \(\mathcal{O}(mn)\) time. If \(F\) is not sequencible, then we report that there is no \PO{} allocation that completes $F$. Otherwise, by Lemma~\ref{lem:extend:frozen}, there is a polynomial time algorithm to compute a \PO{} allocation that completes \(F\).
\end{proof}

\MMSvaluelexicographic*
\begin{proof}
Fix any agent $i\in N$. Let $v_i$ denote the valuation function that is consistent with $i$'s ordinal ranking over the bundles, namely, $v_i(X) > v_i(Y)$ whenever $i$ prefers $X$ over $Y$. Such a valuation function can be constructed by creating an additive valuation function $v_i$ with $v_i(g)=2^{m-r_i(g)}$ for each good $g$ where $r_i(g)$ denotes $g$'s rank under $i$'s preference order over the goods. As noted before, the lexicographic preference extension induces a total order over the bundles in $2^M$. Therefore, the worst bundle in $i$'s maximin partition is the same 
under all consistent cardinal valuations. 

Consider the following recursive algorithm that, given an instance $\langle N , M, \V, F \rangle$, computes an \MMS{} partition of $i\in N$. Let $g_1,g_2,\ldots,g_{\ell}$ be the unallocated goods in $U=M \setminus \cup_{i \in N}F_i$. Suppose without loss of generality that $v_i(F_1) \leq v_i(F_2) \leq \ldots \leq v_i(F_n)$ and that $v_i(g_1) > v_i(g_2) > \ldots > v_i(g_{\ell})$. 

\begin{itemize}
\item If $U=\emptyset$, then return $(F_1,F_2,\ldots,F_n)$.

\item If $U \neq \emptyset$ and $v_i(F_j) \geq v_i(g_{1})$ for every $j \in \{2, \ldots,n\}$, allocate $U$ to $F_1$ and return $(F_1 \cup U, F_2, \ldots,F_n)$.

\item  If $U \neq \emptyset$ and $v_i(F_j) < v_i(g_{1})$ for some $j \geq 2$, apply the algorithm to the restricted instance $I'=\langle N \setminus \{1\}, M \setminus (F_1 \cup \{g_{1}\}), (v_i)_{i \in N \setminus \{1\}}, (F_2,\ldots,F_n) \rangle$ and obtain an allocation $(A_2,\ldots,A_k)$. Return $(F_1 \cup \{g_{1}\}, A_2,\ldots,A_n)$. 
\end{itemize}

We prove by induction on the number of unallocated goods that the algorithm correctly finds an \MMS{} partition of agent $i \in N$. Clearly, if there is no unallocated good, i.e., $U=M \setminus \cup_{i \in N}F_i=\emptyset$, then $(F_1,F_2,\ldots,F_n)$ is an \MMS{} partition of agent $i$ and the algorithm returns such a partition. Assume that the claim holds when $|U|=t$ and consider the case when $|U|=t+1$. 

Let $(C_j)_{j \in [n]}$ be a completion such that $(F_j \cup C_j)_{j \in [n]}$ is an \MMS{} partition of agent $i$, meaning that $\min_{j \in [n]} v_i(F_j\cup C_j)=\mu_i$. Let $(A_1,A_2,\ldots,A_n)$ be a partition returned by the algorithm and ${\hat \mu}_i=\min_{j \in [n]}v_i(A_j)$. We will show that ${\hat \mu}_i =\mu_i$.

First, consider the case when $v_i(F_j) \geq v_i(g_{1})$ for every $j \geq 2$. 
In this case, $A_1=F_1 \cup U$ and $A_j=F_j$ for every $j \geq 2$.  
Since agents have lexicographic valuations and $v_i(F_j) \geq v_i(F_{1})$ for every $j \geq 2$, this means that for every $j \geq 2$, $F_j$ contains a good $f_j$ with $f_j \succ_i g_{1}$ and $f_j \succ_i g$ for every $g \in F_1$. Then, $v_i(F_j) \geq v_i(F_1 \cup U) \geq v_i(F_1 \cup C_1)$, where the second inequality holds because $C_1 \subseteq U$. 
Thus, we have
\begin{align*}
\mu_i &=\min_{j \in [n]} v_i(F_j \cup C_j)\\
      &\leq  \min \{ v_i(F_1 \cup  U), \min_{j \in \{2,\ldots,n\}} v_i(F_j) \} ={\hat \mu}_i. 
\end{align*}
Additionally, we also have that $\mu_i \geq {\hat \mu}_i$ since $(U,\emptyset,\ldots,\emptyset)$ is a completion for instance $I$. Thus, we have $\mu_i = {\hat \mu}_i$. 

Second, consider the case when there is a bundle $F_{j'}$ with $j' \geq 2$ and $v_i(F_{j'}) < v_i(g_{1})$. 
Let $(A'_2,\ldots,A'_n)$ be the partition returned by the algorithm for the instance $I'=\langle N \setminus \{1\}, M \setminus (F_1 \cup \{g_{1}\}), (v_i)_{i \in N \setminus \{1\}}, (F_2,\ldots,F_n) \rangle$. 
By the induction hypothesis, $(A'_2 \setminus F_2,\ldots,A'_n\setminus F_n)$ is a completion for $I'$ and 
\[
\max_{C^* \in \Pi^{n-1}_{U\setminus \{g_{1} \}}} \min_{j \in \{2,\ldots,n\}} v_i(F_j \cup C^*_j)=\min_{j \in \{2,\ldots,n\}} v_i(A'_j),
\]
where $\Pi^{n-1}_{U\setminus \{g_{1} \}}$ denotes the set of all ordered $(n-1)$-partitions of the set $U\setminus \{g_{1} \}$. 
Let $\mu':=\min_{j \in \{2,\ldots,n\}} v_i(A'_j)$. 
By construction of the algorithm, observe that we have ${\hat \mu}_i=\min \{ v_i(F_1 \cup \{g_{1} \}), \mu'\}$.

Now, suppose that $C_1= X \cup \{g_{1} \}$ for some $X \subseteq U$. Then since $v_i(F_{j'} \cup C_{j'}) \leq v_i(F_{j'}\cup C_{j'} \cup X)$ and $v_i(F_{j'} \cup C_{j'}) \leq v_i(F_1 \cup \{g_{1}\})$,  
\[
v_i(F_{j'} \cup C_{j'})\leq \min \{v_i(F_1 \cup \{g_{1}\}), v_i(F_{j'} \cup C_{j'} \cup X)\}. 
\]
This implies
\begin{align*}
\mu_i=&\min_{j \in [n]} v_i(F_j \cup C_j)\\
\leq 
&\min \{ v_i(F_1 \cup  \{g_{1}\}), v_i(F_{j'} \cup C_{j'} \cup X),\min_{j \neq 1,j'} v_i(F_j \cup C_j)\}\\
\leq 
& \min \{ v_i(F_1 \cup \{g_{1}), \max_{C^* \in \Pi^{n-1}_{U\setminus \{g_{1} \}}} \min_{j \in \{2,\ldots,n\}} v_i(F_j \cup C^*_j) \}\\
=
&\min \{ v_i(F_1 \cup \{g_{1}), \mu' \}={\hat \mu}_i.
\end{align*}
The inequality $\mu_i \geq {\hat \mu}_i$ holds since $(\{g_{1}\},A'_2 \setminus F_2,\ldots,A'_n\setminus F_n)$ is a completion for instance $I$. Thus, we have $\mu_i = {\hat \mu}_i$. 

Next, suppose that $g_{1} \not \in C_1$. Let $C_{j''}$ be a bundle containing $g_{1}$. 
Note that since $v_i(C_1) \leq v_i(g_{1})$ and $v_i(F_1) \leq v_i(F_{j''})$, we have $v_i(F_1 \cup C_1) \leq v_i(F_1 \cup \{g_{1}\})$ and $v_i(F_1 \cup C_1) \leq v_i(F_{j''} \cup C_1)$. 
Thus, 
\[
v_i(F_1 \cup C_1) \leq \min \{v_i(F_1 \cup \{g_{1}\}), v_i(F_{j''} \cup C_1)\}. 
\]
This implies  
\begin{align*}
\mu_i=&\min_{j \in [n]} v_i(F_j \cup C_j )\\
\leq 
& \min \{ v_i(F_1 \cup  \{g_{1}\}), v_i(F_{j''} \cup C_1),\min_{j \neq 1, j''}v_i(F_j \cup C_j)\} \\
\leq
& \min \{ v_i(F_1 \cup  \{g_{1}\}), \\ 
& \qquad v_i((F_{j''} \cup C_1 \cup C_{j''}) \setminus \{g_1\}), \min_{j \neq 1, j''}v_i(F_j \cup C_j)\} \\
\leq
& \min \{ v_i(F_1 \cup \{g_{1} \}), \max_{C^* \in \Pi^{n-1}_{U\setminus \{g_{1} \}}} \min_{j \in \{2,\ldots,n\}} v_i(F_j \cup C^*_j) \}\\
= & \min \{ v_i(F_1 \cup \{g_{1} \}), \mu' \}={\hat \mu}_i. 
\end{align*}
Similarly to the above, this implies that $\mu_i= {\hat \mu}_i$. It is immediate that the algorithm can be implemented in polynomial time. 
\end{proof}

To establish Theorem~\ref{thm:decideMMS:lexicographic}, we prove the following lemmas.

\begin{lemma}\label{lem:lexicographic:MMSvalue}
For an instance $\langle N , M, \V, F \rangle$ with lexicographic valuations and agent $i\in N$, let $g_1,g_2,\ldots,g_{\ell}$ be the unallocated goods in $U=M \setminus \cup_{i \in N}F_i$ with $v_i(g_1) > v_i(g_2) > \ldots > v_i(g_{\ell})$. For every $C \subseteq U$ satisfying $v_i(F_i \cup C) \geq \mu_i$, either $C=\emptyset$, or there exists a good $g_t \in C$ with 
\begin{itemize}
\item $v_i(F_i \cup \{ g_t \}) \geq \mu_i$, or
\item $\{g_t,g_{t+1}, \ldots, g_{\ell}\} \subseteq C$ such that $v_i(F_i \cup \{g_t, g_{t+1}, \ldots, g_{\ell}\}) =\mu_i$. 
\end{itemize}
\end{lemma}

\begin{proof}
Consider any agent $i\in N$ and suppose that $v_i(F_1) \leq v_i(F_2) \leq \ldots \leq v_i(F_n)$. 
Let $(C_1, \ldots,C_n)$ be a completion, computed by the algorithm in the proof of Proposition~\ref{prop:computeMMS:lexicographic}, such that $(F_j \cup C_j)_{j \in [n]}$ is an \MMS{} partition of agent $i$. Let $j \in [n]$ be the index of the minimum-value bundles, namely, $v_i(F_j \cup C_j) = \mu_i$. Let $C \subseteq U$ satisfying $v_i(F_i \cup C) \geq \mu_i=v_i(F_j \cup C_j)$

By construction of the algorithm, $C_j$ is either the empty set or consists of $\{g_{t'}, g_{t'+1},\ldots, g_{\ell}\}$ for some $g_{t'} \in U$. Indeed, if $C_j= \{g_t\}$ for some $t< \ell$, then this would mean that $v_i(g_t) > v_i(F_{j'})$ for some $j' > j$, implying that $v_i(F_j \cup C_j) > v_i(F_{j'} \cup C_{j'})$, contradicting the fact that $v_i(F_j \cup C_j) = \mu_i$. 

By the above observation, if $v_i(F_i)=v_i(F_j)$, then $C=C_j$ or $C$ includes a good $g_t$ with $v_i(g_t) > v_i(C_j)$; thus, the claim holds. Suppose that $v_i(F_i) \neq v_i(F_j)$. Consider first the case when $v_i(F_i) < v_i(F_j)$. This means that $\max_{g \in F_i} v_i(g) < \max_{g \in F_j} v_i(g)$. If $\max_{g \in C} v_i(g) < \max_{g \in F_j \cup C_j} v_i(g)$, then we would have $v_i(F_i \cup C) < v_i(F_j \cup C_j)$, a contradiction. Thus, $\max_{g \in C} v_i(g) \geq \max_{g \in F_j \cup C_j} v_i(g)$. 

Now, if $C_j=\emptyset$, then $C$ must contain a good $g_{t} \in U$ such that $v_i(g_t) > \max_{g \in F_j \cup C_j} v_i(g)$, meaning that $v_i(g_t) > v_i(F_j \cup C_)$ and the claim holds.
Suppose that $C_j=\{g_{t'}, g_{t'+1},\ldots, g_{\ell}\}$ for some $g_{t'} \in U$. If $C=C_j$, then $v_i(F_i \cup C) < v_i(F_j \cup C_j)$, a contradiction. Thus, $C \neq C_j$. Thus, $C$ must contain a good $g_{t} \in U$ with $t<t'$ such that $v_i(g_{t}) > \max_{g \in F_j \cup C_j} v_i(g)$, meaning that $v_i(g_t) > v_i(F_j \cup C_j)$ and the claim holds.

Next, consider the case when $v_i(F_i) > v_i(F_j)$. If $C_j=\emptyset$, then any set $S \subseteq U$ satisfies $v_i(F_i \cup S) \geq \mu_i$. Suppose that $C_j=\{g_{t'}, g_{t'+1},\ldots, g_{\ell}\}$ for some $g_{t'} \in U$. Then if $v_i(F_i) > v_i(g_{t'})$, then again any set $S \subseteq U$ satisfies $v_i(F_i \cup S) \geq \mu_i$. If $v_i(F_i) < v_i(g_{t'})$, then by construction of the algorithm this means that $t'=\ell$ and $C$ must contain a good $g_t$ with $t\leq t'$ and thus the claim holds.
\end{proof}

The lemma above establishes that to ensure \MMS{} for each agent, it is enough to allocate either the empty bundle, a single good, or a segment of their bottom-preferred unallocated goods. We then demonstrate that if an \MMS{} allocation exists, there is one with at most one agent of the last type.

\begin{lemma}\label{lem:MMS:lexicographic:special}
For lexicographic valuations, suppose that there are $\ell$ unallocated goods and each agent $i \in N$ has a preference ordering $g^i_1 \succ g^i_2 \succ \ldots \succ g^i_{\ell}$ over the unallocated goods. 
Then, there exists an \MMS{} allocation if and only if there exists a sub-partition $(C_1,C_2,\ldots,C_n)$ of the unallocated goods such that 
\begin{itemize}
\item $(F_1 \cup C_1,\ldots,F_n \cup C_n)$ is an \MMS{} allocation,
\item $|\{\, i\in N \mid |C_i|\geq 2 \,\}| \leq 1$, and
\item if $|C_i|\geq 2$, then there exists an unallocated good $g^i_t$ such that $C_i=\{g^i_{t},g^i_{t+1},\ldots,g^i_{\ell}\}$ and $v_i(F_i \cup C_i)=\mu_i$.
\end{itemize}
\end{lemma}

\begin{proof}
If there exists $(C_1,C_2,\ldots,C_n)$ satisfying the above properties, then clearly, there exists an \MMS{} allocation. Thus, suppose that there exists a sub-partition $(C_1,C_2,\ldots,C_n)$ of the unallocated goods such that $(F_1 \cup C_1,\ldots,F_n \cup C_n)$ is an \MMS{} allocation. Among such completions, we choose the one that minimizes the number of agents $i$ with $|C_i|\geq 2$. Now, if there exists some agent $i \in N$ such that $|C_i| \geq 2$ and $C_i$ contains a good $g$ with $v_i(g) \geq \mu_i$, then we can remove the goods in $C_i$ other than $g$ and decrease the number of agents $i$ with $|C_i|\geq 2$ while still satisfying \MMS{}.
Thus, for every agent $i \in N$ such that $|C_i| \geq 2$, $C_i$ does not contain a good $g$ with $v_i(g) \geq \mu_i$, meaning that by Lemma~\ref{lem:lexicographic:MMSvalue}, $C_i=\{g^i_{t_i},\ldots,g^i_{\ell}\}$ for some $t_i \leq n$. Suppose that there exists a pair of distinct agents $i,j \in N$ with $|C_i| \geq 2$ and $|C_j| \geq 2$. Then, notice that agent $i$ prefers $g^j_{t_j}$ over any goods in $C_i$; likewise, agent $j$ prefers $g^i_{t_i}$ over any goods in $C_j$. Thus, assigning $g^j_{t_j}$ to agent $i$ and $g^i_{t_i}$ to agent $j$ and removing the remaining goods in $C_i$ and $C_j$ still maintains the \MMS{} property, contradicting the minimality of $|\{\, i\in N \mid |C_i|\geq 2 \,\}|$. Hence, we have that $|\{\, i\in N \mid |C_i|\geq 2 \,\}| \leq 1$.   
\end{proof}

We will now prove~\Cref{thm:decideMMS:lexicographic}.

\begin{proof}[Proof of~\Cref{thm:decideMMS:lexicographic}]
Suppose that there are $\ell$ unallocated goods and each agent $i \in N$ has a preference ordering $g^i_1 \succ g^i_2 \succ \ldots \succ g^i_{\ell}$ over the unallocated goods.
By Proposition~\ref{prop:computeMMS:lexicographic}, we can in polynomial time compute the maximin share value $\mu_i$ for each agent $i \in N$. Without loss of generality, assume that $\mu_i > v_i(F_i)$ for each agent $i \in N$; if there is an agent $i \in N$ with $\mu_i \leq v_i(F_i)$, we can remove such an agent from consideration. 
By Lemma~\ref{lem:lexicographic:MMSvalue}, for each agent $i \in N$, it suffices to allocate either a single good, or a segment of their bottom-preferred unallocated goods. Further, by Lemma~\ref{lem:MMS:lexicographic:special}, our problem can be reduced to deciding the existence of an \MMS{} allocation where at most one agent gets multiple goods. 

Formally, for each agent $i\in N$, let $S_i$ be the set of goods $g \in U$ such that $v_i(F_i \cup \{g\}) \geq \mu_i$. Moreover, if there exists a good $g^i_t$ such that $v_i(F_i \cup \{g^i_t, g^i_{t+1}, \ldots, g^i_{\ell}\}) =\mu_i$, we set $M_i=\{g^i_t, g^i_{t+1}, \ldots, g^i_{\ell}\}$, and otherwise, we set $M_i=\emptyset$. Note that by Lemma~\ref{lem:lexicographic:MMSvalue} and by the assumption that $\mu_i > v_i(F_i)$ for each agent $i \in N$, $S_i \neq \emptyset$ or $M_i \neq \emptyset$ for each agent $i \in N$. 

Now we first determine the existence of an \MMS{} allocation where each agent obtains a single good in $S_i$. To this end, create a bipartite graph with bipartition $(N,M)$ and there is an edge between an agent $i \in N$ and a good $g \in M$ if and only if $g \in S_i$. We can in polynomial time find out whether there exists a matching covering the set of agents $N$. If there exists such a matching, then return an allocation that assigns to each agent $i \in N$ the single good allocated to agent $i$ in the matching as well as $F_i$. 

Otherwise, we proceed to checking whether for each agent $i \in N$ with $M_i \neq \emptyset$, there exists a sub-partition of the unallocated goods such that agent $i$ obtains $M_i$ and each of the remaining agent obtains a single good in $S_i$. Again, to do so, for each agent $i \in N$ with $M_i \neq \emptyset$, we create a bipartite graph $(N \setminus \{i\},M)$ where there is an edge between an agent $i' \in N \setminus \{i\}$ and a good $g \in M$ if and only if $g \in S_{i'}$, and decide the existence of a matching covering $N \setminus \{i\}$. This can be done in polynomial time. If there exists such a matching, then return an allocation that assigns to agent $i$ the bundle $F_i \cup M_i$ and to each agent $i' \in N \setminus \{i\}$ the single good allocated to agent $i'$ in the matching as well as $F_{i'}$. 

It can be easily verified that if the algorithm returns an allocation, the resulting allocation is \MMS{} by construction. 
Conversely, if there exists an \MMS{} allocation, by Lemma~\ref{lem:MMS:lexicographic:special}, there exists a sub-partition of the unallocated goods $(C_1,C_2,\ldots,C_n)$ satisfying the three properties of Lemma~\ref{lem:MMS:lexicographic:special}. Thus, the algorithm finds out at least one \MMS{} allocation. 
\end{proof}

\section{Omitted Material from Section~\ref{sec:additive}}\label{appendix:additive}

\EFoneTwoAgentAdditiveNPhard*
\begin{proof}
\noindent
{\bf \EF{1}}: Whether a given allocation correctly completes a frozen allocation and satisfies \EF{1} can be checked in polynomial time. Thus, the problem is clearly in \NP{}. 

To show NP-hardness, we present a reduction from \Partition{}. Recall that given positive integers $w_1,w_2,\ldots, w_m$ such that $\sum_{i=1}^{m} w_i = 2T$, \Partition{} asks whether there is a set $S \subseteq [m]$ such that $\sum_{i \in S} w_i = \sum_{i \in [m] \setminus S} w_i = T$. Without loss of generality, we can assume that for each $i \in [m]$, $w_i \leq T$.

Given an instance of \Partition{}, we create an instance of the completion problem as follows. 
There are two agents, $1$ and $2$. 
There are two frozen goods $f_1$ and $f_2$ where agent $1$ receives $f_1$ and agent $2$ receives $f_2$. Agent $1$ has value $0$ for $f_1$ and $T$ for $f_2$. Agent $2$ has value $0$ for $f_2$ and $T$ for $f_1$. 
For each $k \in [m]$, we create a good $k$ that each agent values at $w_k$. Denote the frozen allocation by $(F_1,F_2)$. 

We show that there exists an \EF{1} allocation that completes the frozen allocation if and only if the original \Partition{} problem has a yes-instance, i.e., a partition of the unallocated goods into two bundles of value $T$. First, suppose that there exists an \EF{1} allocation $A = (F_i \cup C_i)_{i \in \{1,2\}}$ that completes the frozen allocation, where $(C_1,C_2)$ is a completion. 
Then, the frozen good $f_1$ is the most valuable good among the goods allocated to agent $1$ from agent $2$'s perspective since every good has a value at most $T$ according to $2$'s valuation. Therefore, $v_2(A_2)=v_2(C_2) \geq v_2(C_1) =v_2(A_2) - T$. Similarly, we have $v_1(C_1) \geq v_1(C_2)$. This means that $\sum_{j  \in C_1}w_j = \sum_{j  \in C_2} w_j=T$. Thus, $(C_1,C_2)$ is a desired partition. 

Conversely, suppose that \Partition{} problem has a yes-instance, namely, there exists a partition $(S,[m] \setminus S)$ of $[m]$ such that $\sum_{i \in S} w_i = \sum_{i \in [m] \setminus S} w_i = T$. Complete the frozen allocation $F$ by assigning the goods in $S$ to agent $1$ and the goods in $[m] \setminus S$ to agent $2$. It can be easily checked that the resulting allocation is \EF{1}. 

\noindent
{\bf \Prop{1}}: 
Checking whether a given allocation correctly completes a given frozen allocation and whether it satisfies \Prop{1} can be done in polynomial time. Thus, the problem is clearly in \NP{}. To show NP-hardness, we present a reduction from \textsc{PARTITION}.

Given an instance of \textsc{PARTITION}, a set of positive integers $\{w_1,w_2,\cdots,w_m\}$ such that $\sum_{i\in [m]} w_i = 2T$, we create an instance of the completion problem as follows. 
There are two agents $1$ and $2$. 
There are four frozen goods $f_1, f_2, f_3$, $f_4$ where agent $1$ receives $f_1,f_2$ and agent $2$ receives $f_3,f_4$. Agent $1$ has value $0$ for each of $f_1,f_2$ and $T$ for each of $f_3,f_4$. Agent $2$ has value $0$ for each of $f_3,f_4$ and $T$ for each of $f_1,f_2$. 
For each $k \in [m]$, we create a good $k$ that each agent values at $w_k$. Denote by $(F_i)_{i \in \{1,2\}}$ the frozen allocation. The proportionality threshold for each agent is $\frac{2T+2T}{2} = 2T$. 

We show that there exists an \Prop{1} allocation that completes the frozen allocation if and only if the original \textsc{PARTITION} problem has a yes-instance, i.e., a partition of the unallocated goods into two bundles of value $T$. First, suppose that there exists an \Prop{1} allocation $A = (F_i \cup C_i)_{i =1,2}$ that completes the frozen allocation, where $(C_i)_{i=1,2}$ is a completion. 
Then, the frozen good $f_1$ is the most valuable good among the goods allocated to agent $1$ from agent $2$'s perspective since every good has value at most $T$ according to $2$'s valuation. Therefore, $v_2(A_2)=v_2(C_2) \geq T$. Similarly, we have $v_1(C_1) \geq T$. This means that $\sum_{j  \in C_1}w_j = \sum_{j  \in C_2} w_j=T$. Thus, $(C_1,C_2)$ is a desired partition. 

Conversely, suppose that \textsc{PARTITION} problem has a yes-instance, namely, there exists a partition $(S,[m] \setminus S)$ of $[m]$ such that $\sum_{i \in S} w_i = \sum_{i \in [m] \setminus S} w_i = T$. Complete the frozen allocation $F$ by assigning the goods in $S$ to agent $1$ and the goods in $[m] \setminus S$ to agent $2$. It can be easily checked that the resulting allocation is \Prop{1}. 
\end{proof}

\EFoneThreeAgentIdenticalAdditiveNPhard*
\begin{proof}
\noindent
{\bf \EF{1}}: The problem is clearly in \NP{}. To show NP-hardness, we present a reduction from \Partition{}. 
Given an instance of \Partition{}, a set of positive integers $\{w_1,w_2,\cdots,w_m\}$ such that $\sum_{i\in [m]} w_i = 2T$, we create an instance of the completion problem 
as follows. There are three agents $1,2,3$ with the identical valuation $v$. There are two frozen goods $f_1$ and $f_2$ of value $T$ allocated to agent $1$. For each $k \in [m]$, we create a good $k$ with the value $w_k$. 

We show that there exists an \EF{1} allocation that completes the frozen allocation if and only if the original \Partition{} problem has a yes-instance. Suppose there exists an \EF{1} allocation $A = (F_i \cup C_i)_{i \in \{a,b\}}$ that completes the frozen allocation, where $(C_i)_{i \in \{a,b\}}$ is a completion. 
Then, in order for agent $2$ not to have \EF{1}-envy towards agent $1$, agent $2$ needs to receive a bundle of value at least $T$, i.e., $v(A_2)=v(C_2) \geq T$. Similarly, agent $3$ must receive a bundle of value at least $T$ in the completion, i.e., $v(C_3) \geq T$. This means that $\sum_{j  \in C_2}w_j = \sum_{j  \in C_3} w_j=T$. Thus, $(C_2,C_3)$ is a desired partition. 

Conversely, suppose that \Partition{} problem has a yes-instance, namely, there exists a partition $(S,[m] \setminus S)$ of $[m]$ such that $\sum_{i \in S} w_i = \sum_{i \in [m] \setminus S} w_i = T$. Complete the frozen allocation $F$ by assigning the goods in $S$ to agent $2$ and the goods in $[m] \setminus S$ to agent $3$. It can be easily checked that the resulting allocation is \EF{1}.

\noindent
{\bf \Prop{1}}: 
Checking whether a given allocation correctly completes a given frozen allocation and whether it satisfies \Prop{1} can be done in polynomial time. Thus, the problem is clearly in \NP{}. To show NP-hardness, we present a reduction from \textsc{PARTITION}. 

We show a reduction from \textsc{PARTITION}. Given an instance of \textsc{PARTITION}, a set of positive integers $\{w_1,w_2,\cdots,w_m\}$ such that $\sum_{i\in [m]} w_i = 2T$, we create an instance of our problem as follows. 
There are three agents $1,2, 3$ with identical valuation $v$. 
Let $\alpha \ge \max_{j \in [m]}w_j$. Create two frozen goods $f$ and $g$ of value $\alpha$ where $f$ is allocated to agent $2$ and $g$ is allocated to agent $3$. Let $\epsilon \leq \min_{j \in [m]}w_j$ with $ \ell \cdot \epsilon = T+4\alpha$ for some positive integer $\ell$ and create $\ell$ frozen goods of value $\epsilon$ allocated to agent $1$. Denote by $(F_i)_{i=1,2}$ the frozen allocation. 
For each $k \in [m]$, we create an unallocated good $k$ of value $w_k$. 
The proportionality threshold for each agent is $\frac{3T+4\alpha + 2\alpha}{3} = T+2\alpha$. See Figure~\ref{fig:reduction} for an illustration of our frozen allocation.

\tikzset{every picture/.style={line width=0.75pt}} 

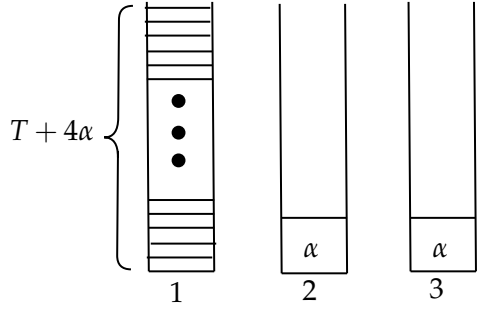
\begin{figure}
    \centering
    
\begin{tikzpicture}[x=0.75pt,y=0.75pt,yscale=-1,xscale=1]
\draw    (246,5) -- (247,141) ;
\draw    (247,141) -- (280,141) ;
\draw    (279,5) -- (280,141) ;
\draw    (313,6) -- (314,142) ;
\draw    (314,142) -- (347,142) ;
\draw    (346,6) -- (347,142) ;
\draw    (378,5) -- (379,141) ;
\draw    (379,141) -- (412,141) ;
\draw    (411,5) -- (412,141) ;
\draw    (314,114) -- (347,114) ;
\draw    (379,114) -- (412,114) ;
\draw    (246,134) -- (279,134) ;
\draw    (248,127) -- (281,127) ;
\draw    (247,119) -- (280,119) ;
\draw    (247,112) -- (280,112) ;
\draw    (247,105) -- (280,105) ;
\draw    (246,44) -- (279,44) ;
\draw    (246,37) -- (279,37) ;
\draw    (246,30) -- (279,30) ;
\draw    (245,22) -- (278,22) ;
\draw    (245,15) -- (278,15) ;
\draw    (245,8) -- (278,8) ;
\draw  [fill={rgb, 255:red, 0; green, 0; blue, 0 }  ,fill opacity=1 ] (259,55) .. controls (259,53.34) and (260.34,52) .. (262,52) .. controls (263.66,52) and (265,53.34) .. (265,55) .. controls (265,56.66) and (263.66,58) .. (262,58) .. controls (260.34,58) and (259,56.66) .. (259,55) -- cycle ;
\draw  [fill={rgb, 255:red, 0; green, 0; blue, 0 }  ,fill opacity=1 ] (259,71) .. controls (259,69.34) and (260.34,68) .. (262,68) .. controls (263.66,68) and (265,69.34) .. (265,71) .. controls (265,72.66) and (263.66,74) .. (262,74) .. controls (260.34,74) and (259,72.66) .. (259,71) -- cycle ;
\draw  [fill={rgb, 255:red, 0; green, 0; blue, 0 }  ,fill opacity=1 ] (259,85) .. controls (259,83.34) and (260.34,82) .. (262,82) .. controls (263.66,82) and (265,83.34) .. (265,85) .. controls (265,86.66) and (263.66,88) .. (262,88) .. controls (260.34,88) and (259,86.66) .. (259,85) -- cycle ;
\draw   (239,7) .. controls (234.33,6.97) and (231.98,9.28) .. (231.95,13.95) -- (231.58,63.45) .. controls (231.53,70.12) and (229.17,73.43) .. (224.5,73.39) .. controls (229.17,73.43) and (231.48,76.78) .. (231.43,83.45)(231.45,80.45) -- (231.05,132.95) .. controls (231.02,137.62) and (233.33,139.97) .. (238,140) ;

\draw (323.5,125) node [anchor=north west][inner sep=0.75pt]   [align=left] {$\displaystyle \alpha $};
\draw (388.5,125) node [anchor=north west][inner sep=0.75pt]   [align=left] {$\displaystyle \alpha $};
\draw (175,63) node [anchor=north west][inner sep=0.75pt]   [align=left] {$\displaystyle T+4\alpha $};
\draw (256,145) node [anchor=north west][inner sep=0.75pt]   [align=left] {$\displaystyle {1}$};
\draw (323,144) node [anchor=north west][inner sep=0.75pt]   [align=left] {$\displaystyle {2}$};
\draw (387,143) node [anchor=north west][inner sep=0.75pt]   [align=left] {$\displaystyle {3}$};

\end{tikzpicture}
\caption{Three agents with fixed allocation $T+7\alpha, \alpha \text{ and } \alpha$}
    \label{fig:reduction}
\end{figure}

First, suppose that there exists a \Prop{1} allocation $A = (F_i \cup C_i)_{i =a,b}$ that completes the frozen allocation $F$, where $(C_i)_{i=a,b}$ is a completion. 
Observe that the maximum value of a good outside agent $2$'s bundle is $\alpha$, namely, $\max_{g \in M \setminus A_2}v(g)=\alpha$. Thus, in order to meet \Prop{1} for agent $2$, agent $2$ needs to receive a bundle of value at least $T$ in the completion because $v(A_2)=v(C_2) + \alpha \geq T + \alpha$, meaning that $v(C_2) \geq T$. Similarly, agent $3$ must receive a bundle of value at least $T$ in the completion, i.e., $v(C_3) \geq T$. This means that $\sum_{j  \in C_2}w_j = \sum_{j  \in C_3} w_j=T$. Thus, $(C_2,C_3)$ is a desired partition.

Conversely, suppose that there exists a partition $(S,[m] \setminus S)$ of $[m]$ such that $\sum_{i \in S} w_i = \sum_{i \in [m] \setminus S} w_i = T$. Then it is easy to see that allocating bundle $S$ to agent $2$ and $[m] \setminus S$ to agent $3$ yields a \Prop{1} allocation.~\qedhere
\end{proof}

\EFonetwo*
\begin{proof}
\noindent
{\bf \EF{1}}: Suppose there are two agents, $1$ and $2$, with identical valuation $v$. 
When agents have identical valuations, the envy graph is always acyclic. If the frozen allocation $F$ is \EF{1}, we can compute an \EF{1} allocation that extends $F$ using Proposition~\ref{thm:ef1_acyclic}. Consider the case when $F$ is not \EF{1}. Without loss of generality, suppose agent $1$ has \EF{1}-envy towards agent $2$, i.e., $F_2$ is not empty and $v(F_1) < v(F_2 \setminus \{g\})$ for every $g \in F_2$. Allocate goods in $U=M \setminus \cup_{i=1,2}F_i$ to agent $1$ till agent $1$ stops having \EF{1}-envy towards agent $2$, i.e., $v(A_1) \geq \min_{g \in A_2} v(A_2 \setminus \{g\})$. 
If we exhaust all the goods while agent $1$ has \EF{1}-envy towards agent $2$, then there is no \EF{1} complete allocation that completes $F$. If agent $1$ stops having \EF{1}-envy towards agent $2$ at any point, either agent $2$ does not envy agent $1$ or the new envy from agent $2$ to agent $1$ can be eliminated by removing the last good allocated to agent $1$. This allocation can be completed using Proposition~\ref{thm:ef1_acyclic}.\\   

\noindent
{\bf \Prop{1}}: Suppose there are two agents, $1$ and $2$, with identical valuation $v$. Recall that for identical valuations, the envy graph with respect to the frozen allocation $F$ is acyclic. 

Consider the case when the frozen allocation $F$ is \Prop{1}. Since the valuations are identical, there cannot be an envy cycle. Therefore,  without loss of generality, we assume that 
agent $2$ does not envy agent $1$ under the frozen allocation. To assign the unallocated goods $U=M\setminus \cup_{i \in N}F_i$, we let each agent $i \in \{1,2\}$ pick its favorite remaining good in $U$ in a round robin fashion, starting with agent \(1\). We claim that the resulting complete allocation $A=(F_i \cup C_i)_{i \in N}$ is \Prop{1}. 

Agent~$1$ does not envy $2$ under the completion $C$, i.e., $v(C_1) \geq v(C_2)$, and the frozen allocation $F$ satisfies \Prop{1} for agent $1$, i.e., $v(F_1) + \max_{j \in F_2} v(\{j\}) \geq \frac{1}{2}(v(F_1)+v(F_2))$. 
Thus, we have that 
\begin{align*}
    v(C_1 \cup F_1)  & + \max_{j \in F_2} v(\{j\}) \\
    & = v(C_1) + v(F_1) + \max_{j \in F_2} v(\{j\}) \\
    & \geq  \frac{1}{2} \left( v(C_1)+v(C_2)+v(F_1)+v(F_2) \right),
\end{align*}
which implies that $A$ is \Prop{1} for agent $1$. 
Now, consider agent $2$. We know that agent $2$ does not envy agent $1$ under the frozen allocation, and it can be observed that agent $2$'s envy toward agent $1$ is bounded up to one good in the completion $C$. Thus, the final allocation $A$ is \EF{1} and hence \Prop{1} for agent $2$.

Now consider the case when the frozen allocation $F$ is not \Prop{1}. Without loss of generality, suppose \Prop{1} is violated from agent $1$'s perspective. (Note that due to identical valuations, \Prop{1} violation can occur for at most one of the two agents.) That is, $v(F_1) + \max_{j \in F_2} v(\{j\}) < \frac{1}{2}v(F_1 \cup F_2)$.  Consider assigning the unallocated goods in $U=M \setminus \cup_{i=1,2}F_i$ to agent $1$ until the partial allocation is \Prop{1} for agent $1$. If we exhaust all the goods without achieving \Prop{1} for agent $1$, then there is no \Prop{1} complete allocation that completes $F$. Otherwise, if the partial allocation becomes \Prop{1} for agent $1$ at some stage, we can treat the current partial allocation as the new `frozen' allocation, and complete it using the procedure described in the previous case. The final allocation must again be \Prop{1}.
\end{proof}

\begin{restatable}{theorem}{EFonestronglyNPC}\label{thm:ef1_strongly_NPC}
For a general number of agents with identical additive valuations, \EFoneCompletion{} and \ProponeCompletion{} are strongly \NPC{}.
\label{thm:EFone_strongly_NPC}
\end{restatable}
\begin{proof}
    {\bf \EF{1}}: The problem is clearly in \NP{}. To show strong \NP{}-hardness, we present a reduction from \textsc{3-Partition}. An instance of \textsc{3-Partition} is given by a set of $3n$ numbers $w_1,w_2,\dots,w_{3r}$ where $r \in \mathbb{N}$. Let $T$ be such that $\sum_{i \in [3r]} w_i = rT$. The goal is to determine if there exists a partition of $w_1,w_2,\dots,w_{3r}$ into $r$ triples such that the sum of numbers in each triple is $T$. This problem is known to be \NPC{} even when, for every~$i \in [3r]$, we have $w_i \approx T/3$. Specifically, there is a fixed constant $\delta > 0$ such that for every~$i \in [3r]$, we have $w_i \in [T/3 - \delta, T/3 + \delta]$, see \cite{strong3partition}.

    Given any instance of \textsc{3-Partition}, we will create an instance of the completion problem as follows. There are $r+1$ agents, namely $1,2,\dots,r+1$, with the same valuation function $v$. There are two frozen goods $f_1$ and $f_2$, each valued at $T$ by all agents, that are allocated to agent $r+1$. For every $k \in [3r]$, we create a good $g_k$ with the value $w_k$ for all agents.

    We will now argue that the given \textsc{3-Partition} instance admits the desired partition into triples if and only if an \EF{1} completion exists in the fair division instance. For the proof of forward direction, let's assume that there exists a partition into triples $S_1,S_2,\dots,S_r$ such that the sum of numbers in each triple is $T$. For each triple $S_i$, we will assign the goods corresponding to the numbers in $S_i$ to agent $i$. Thus, for every $i \in [r]$, agent $i$'s value for its bundle is $T$. Additionally, for any $k \neq r+1$, agent $i$ values agent $k$'s bundle is also $T$, while its value for agent $(r+1)$'s bundle after removing the frozen good $f_1$ is $T$. It is easy to see that agent $(r+1)$ does not envy any other agent. Thus, the completion is \EF{1}.

    In the reverse direction, suppose an \EF{1} completion exists. Let us denote the final allocation by $A=(F_i \cup C_i)_{i \in [r+1]}$. Since $A$ is \EF{1} and the valuations are identical, each agent in $\{1,\dots,r\}$ must receive a value of at least $T$, as otherwise such agents would envy agent \(r+1\) even after removal of any good. The total value of the unfrozen goods is $rT$, which means that all unfrozen goods must be allocated to the agents in $\{1,\dots,r\}$. Furthermore, each agent in $\{1,\dots,r\}$ must receive a value of exactly $T$. To obtain a value of $T$, each agent must receive at least three goods (recall that each good's value is nearly $T/3$). This assignment of goods naturally induces a valid partition into triples of the integers $w_1,\dots,w_{3r}$.\\

    {\bf \Prop{1}}: To show strong \NP{}-hardness of \ProponeCompletion{}, we will again present a reduction from \textsc{3-Partition}. We will assume that the numbers $w_1,\dots,w_{3r}$ in \textsc{3-Partition} are such that for every $i \in [3r]$, we have $w_i > > 1$; this property can be ensured by appropriate scaling.

    Given an instance $w_1,\dots,w_{3r}$ of \textsc{3-Partition}, we will create a fair division instance with $r+1$ agents, namely $1,2,\dots,r+1$, with identical valuations. There are $1 + (r+1)T$ frozen goods, all of them allocated to agent \(r+1\), consisting of one good of value $T$ and $(r+1)T$ other goods each of value \(1\). Additionally, there are $3r$ unfrozen goods corresponding to the $3r$ integers $w_1,\dots,w_{3r}$. The $i^\textup{th}$ such good is valued at $w_i$.

    Let us first present the argument for the forward direction. Note that the proportionality threshold in the fair division instance is $\theta \coloneqq \frac{rT + T + (r+1)T}{r+1} = 2T$. Given a desired partition into triples $S_1,\dots,S_r$ in the \textsc{3-Partition} instance, we create a \Prop{1} allocation as follows: Assign the goods corresponding to the numbers in $S_i$ to agent $i$. This gives each agent in $\{1,\dots,r\}$ a utility of $T$. Each such agent can reach the proportionality threshold $\theta = 2T$ by hypothetically including the frozen good valued at $T$ owned by agent $(r+1)$. Since agent $(r+1)$ already satisfies proportionality, the resulting allocation is \Prop{1}.

    In the reverse direction, suppose a \Prop{1} allocation exists. Let us denote this allocation by $A=(F_i \cup C_i)_{i \in [r+1]}$. We claim that each agent in $\{1,\dots,r\}$ must derive a utility of at least $T$ under $A$. This is because the proportionality threshold is $\theta = 2T$, and the largest good available to any agent in $\{1,\dots,r\}$ for hypothetical addition under \Prop{1} is of value $T$ (this is the frozen good owned by agent $r+1$). By a similar argument as in the \EF{1} reduction above, it can be observed that any such allocation induces the desired partition into triples in the given instance of \textsc{3-Partition}.    
\end{proof}

\end{document}